\documentclass[10pt,journal]{IEEEtran}
\usepackage[pdftex]{graphicx}
\graphicspath{{../pdf/}{../jpeg/}}
\DeclareGraphicsExtensions{.pdf,.jpeg,.png}

\usepackage[table]{xcolor}
\usepackage{soul}
\usepackage{framed}
\usepackage{graphicx}
\usepackage{lipsum}
\usepackage{float}
\usepackage{cuted}
\colorlet{shadecolor}{yellow}
\usepackage[cmex10]{amsmath}
\usepackage{array}
\usepackage{mdwmath}
\usepackage{amsmath}
\usepackage{textcase}

\usepackage{colortbl}
\usepackage{makecell}
\usepackage{multirow}
\usepackage{multicol}

\usepackage{subcaption}
\usepackage{algorithm}
\usepackage{algorithmic}
\usepackage{diagbox}  
\usepackage{mdwtab}
\usepackage{eqparbox}
\usepackage{url}
\usepackage{caption}
\usepackage{amsthm,amsfonts,amssymb,amsmath,graphicx}

\newtheorem{Theorem}{Theorem}

\ifCLASSOPTIONcompsoc
  \usepackage[nocompress]{cite}
\else
  \usepackage{cite}
\fi
\ifCLASSINFOpdf
\else
\fi

\begin{document}
\captionsetup{font=small} 
\title{StreamOptix: A Cross-layer Adaptive Video Delivery Scheme}

\author{{Mufan~Liu,~\IEEEmembership{Student Member,~IEEE,}
        Le~Yang,~\IEEEmembership{Member,~IEEE,}
        Yifan~Wang,
        Yiling~Xu,~\IEEEmembership{Member,~IEEE,}
        Ye-Kui~Wang,
        Yunfeng~Guan
        }% <-this % stops a space
% \IEEEcompsocitemizethanks{\IEEEcompsocthanksitem M. Shell was with the Department
% of Electrical and Computer Engineering, Georgia Institute of Technology, Atlanta,
% GA, 30332.\protect\\
% % note need leading \protect in front of \\ to get a newline within \thanks as
% % \\ is fragile and will error, could use \hfil\break instead.
% E-mail: see http://www.michaelshell.org/contact.html
% \IEEEcompsocthanksitem J. Doe and J. Doe are with Anonymous University.}% <-this % stops a space
\thanks{This is an extended version of the one appeared at GlobeCom 2023 \cite{globecom}.}}


\markboth{Journal of \LaTeX\ Class Files,~Vol.~14, No.~8, August~2015}%
{Shell \MakeLowercase{\textit{et al.}}: Bare Advanced Demo of IEEEtran.cls for IEEE Computer Society Journals}



\IEEEtitleabstractindextext{%
\begin{abstract}
This paper presents a cross-layer video delivery scheme, StreamOptix, and proposes a joint optimization algorithm for video delivery that leverages the characteristics of the physical (PHY), medium access control (MAC), and application (APP) layers. Most existing methods for optimizing video transmission over different layers were developed individually. Realizing a cross-layer design has always been a significant challenge, mainly due to the complex interactions and mismatches in timescales between layers, as well as the presence of distinct objectives in different layers. To address these complications, we take a divide-and-conquer approach and break down the formulated cross-layer optimization problem for video delivery into three sub-problems. We then propose a three-stage closed-loop optimization framework, which consists of 1) an adaptive bitrate (ABR) strategy based on the link capacity information from PHY, 2) a video-aware resource allocation scheme accounting for the APP bitrate constraint, and 3) a link adaptation technique utilizing the soft acknowledgment feedback (soft-ACK). The proposed framework also supports the collections of the distorted bitstreams transmitted across the link. This allows a more reasonable assessment of video quality compared to many existing ABR methods that simply neglect the distortions occurring in the PHY layer. Experiments conducted under various network settings demonstrate the effectiveness and superiority of the new cross-layer optimization strategy. A byproduct of this study is the development of more comprehensive performance metrics on video delivery, which lays down the foundation for extending our system to multimodal communications in the future. Code for reproducing the experimental results is available at https://github.com/Evan-sudo/StreamOptix.
\end{abstract}

}
\maketitle

% is not selected <OR> if conference mode is selected - because compsoc
% conference papers position the abstract like regular (non-compsoc)
\IEEEdisplaynontitleabstractindextext

% \fi
\IEEEpeerreviewmaketitle

% \ifCLASSOPTIONcompsoc
% \IEEEraisesectionheading{\section{Introduction}\label{sec:introduction}}
% \else

\section{Introduction}
The deployment of 5G technology promotes the rapid proliferation of mobile video streaming.
% The advancements in network capabilities lead to a considerable increase in video traffic, coupled with growing user expectations for high-quality viewing experiences. 
According to Ericsson's annual report 2022 \cite{Ericsson}, it is anticipated that the volume of mobile data traffic will double between 2023 and 2027, with an annual growth rate of nearly 30\%. On the other hand, this poses signiﬁcant challenges to the development of efficient methods for video delivery, due to the increase in the complexity of wireless transmission scenarios and higher user expectations for viewing experiences.

To attain satisfactory quality of experience (QoE) in a fast-changing environment, methods for optimizing video delivery have been proposed. Among them, more than half were built on Dynamic Adaptive Video Streaming over HTTP (DASH) \cite{dash}, which serves as the backbone for realizing various video delivery strategies in the last decade. In DASH, videos are first segmented and encoded into chunks with different bitrates and resolutions, more commonly known as representations, and are then stored on the server. When the video content is requested, the DASH server considers both playback and network conditions such as buffer size and throughput\footnote{Throughout this paper, we say \textbf{throughput} when talking about network bandwidth and \textbf{bitrate} when talking about encoding quality.}, and switches among the available chunk representations to ensure good QoE. This process, referred to as the adaptive bitrate (ABR), enables the system to adapt to the fluctuating network conditions and guarantees a seamless playback experience. In real world, video transmission encompasses not just the upper-layer ABR but also the complex lower-layer wireless link, commonly referred to as the 'last mile' transmission \cite{mmcapacity}. The large gap between these two layers has led to previous studies on video delivery treating them separately. This separation has resulted in an optimization bottleneck in current video delivery research.

\begin{figure}
\centering
\includegraphics[width=\linewidth]{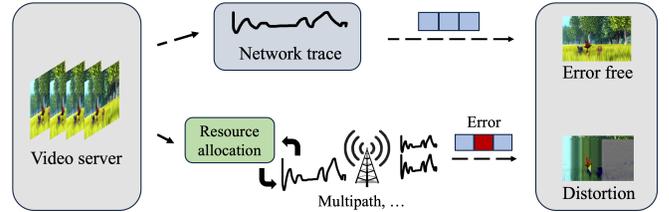}
\caption{Trace-driven ABR environment (upper) overlooks the effect of resource allocation and channel fading over wireless link (bottom).}
\label{fig:images}
\vspace{-0.4cm}
\end{figure}

\subsection{Motivation}
ABRs are currently evaluated in a way at the application layer (APP) without any transmission error, known as trace-driven simulation. This method simulates video streaming by using precollected throughput traces. It updates the video downloading and playback processes effectively  based on the correlation between the video size and throughput trace. However, trace-driven environment overlooks the impact of the wireless link, despite its simplicity and speed. As a result, real world video streaming entails intricate wireless link transmissions, leading to distortions caused by multipath fading and fluctuations in throughput due to resource allocation. (Fig. 1). To emphasize its effects, we constructed a 5G physical downlink shared channel (PDSCH) to test existing DASH-bassed ABRs and analyzed the decoded bitstreams\footnote{
The experimental setup was based on the static configuration described in Section 4, but retained the default settings for link adaptation and resource allocation.}. As depicted in Fig. \ref{fig:images}, resource allocation modifies throughput during streaming, which is unobservable in a trace-driven environment. Subsequently, the video displays arbitrary distortions following the fading of the wireless link. To sum up, the evaluation revealed two issues. First, Fig. 2a shows that \textbf{higher QoE scores, estimated by ABRs based on the selected birates, do not always correspond to better received visual quality} (in this experiment, the structural similarity (SSIM) quality metric was used). In other words, there exists \textit{inconsistency} between the ABR-predicted QoE and actually received visual quality scores post wireless transmission. 
% The reason is that current ABRs assume error-free transmission of the requested videos, which may be invalid for a wireless fading channel. 
Neglecting the transmission errors in this case leads to overestimated user QoE by ABRs. 

Another observation is illustrated in Fig. 2b. We noticed \textbf{an evident \textit{gap} between the selected bitrates and instantenous network capacity}. Due to the inherently ill-posed nature of throughput prediction, it is inevitable that the video bitrate based on predicted throughput will deviate from the actual throughput (see e.g. \cite{bola, mpc, pensieve, festive}).
% Most work attribute this gap to the errors in the estimates of the throughput obtained at the APP layer (see e.g. \cite{bola, mpc, pensieve, festive}), because throughput prediction is inherently an ill-posed problem.
Improved throughput prediction algorithms were proposed in \cite{payloadaware, situ}, but they did not utilize that the throughput can in fact be recontrolled after prediction. In fact, incorporating video bitrate considerations into resource allocation to control throughput can effectively reduce this gap.
% In particular, the PHY's link adaptation modifies the size of each transport block (TB), while MAC can impact throughput by dynamically allocating transmission resources, thereby influencing the throughput. 

Traditionally, link control conducted at PHY and MAC was decoupled from the APP, focusing on optimizing the Quality of Service (QoS) of the network only while ignoring the video-aware information \cite{mmcapacity, OFDM, crossbmsb, Prius, access}. Besides, the ABR mechanisms in APP usually lack the use of feedbacks from the PHY and MAC layers, resulting in inaccurate evaluation of the QoE, as well as limited performance gain. % This leads to a situation where what ought to be closed-looped seems disjointed and separated. 
% There has been an increasing number of cross-layer methods proposed in recent literature. In particular, a straightforward way for realizing cross-layer designs is to employ measurements from wireless links as feedbacks to the ABR. In \cite{pi, crossbmsb}, it is suggested to replace the chunk-level throughput estimation with cellular throughput measured at the user's device or the estimated channel state information (CSI). 
% \cite{xiong} reformulates the resource allocation problem in MAC by taking into account the video content complexity, and maps users' QoS into QoE for QoE maximization. \cite{appslice} considers the video bitrate as an extra known constraint during resource allocation. However, it does not incorporate buffer's dynamics. 
Existing cross-layer studies either followed a top-down or bottom-up approach in their designs \cite{BMSB2, mmcapacity, xiong, OFDM, access}. However, the complex interplay between layers was not fully exploited, and a closed-loop cross-layer technique is yet to be developed. Moreover, these work did not construct and provide a practical and efficient cross-layer video delivery platform. These motivate taking a cross-layer approach that integrates ABR with link adaptation in PHY and resource allocation in MAC to combat link distortions and improve bitrate utilization jointly.
\begin{figure}
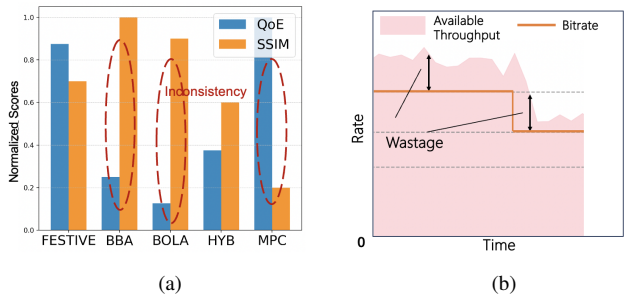

\vspace{-0.3cm}
  \begin{subfigure}{0.499\columnwidth}
    \includegraphics[height=0.78\linewidth]{photo/discovery.pdf}
    \caption{}
  \end{subfigure}
\begin{subfigure}{0.476\columnwidth}
    \includegraphics[height=0.8\linewidth]{photo/resource.pdf}
    \caption{}
  \end{subfigure}
  \caption{Evaluation of existing ABR strategies under simulated 5G wireless links. (a) Inconsistency between the predicted QoE and actually received SSIM. (b) Gap between the selected video bitrate and link capacity. }
  \label{observation}
  \vspace{-0.3cm}
\end{figure}

\subsection{Our design}
This paper presents a three-level closed-loop cross-layer framework for optimizing video delivery. Our method integrates the PHY, MAC, and APP layers.
% APP video bitrate by dynamically choosing the most appropriate bitrate for chunk download, guaranteeing the best possible quality and fluidity. MAC assigns RBs during transmission to encourage a more efficient utilization of throughput. PHY is dedicated to link adaptation to eradicate transmission errors and uphold a consistent throughput. We do not include the transport and network layers as the transport layer is more related to live streaming while our focus is on Video-on-demand, and the network layer poses significant difficulties and additional work for network topology simulation.
In the APP, we employ model predictive control (MPC) for rate adaptation based on the link capacities provided by PHY. The throughput measured in APP is no longer used, because it is an irregularly sampled sequence (ISS) due to the fluctuations in the chunk downloading time. In contrast, the intervals for estimating the link capacities at PHY are much smaller \cite{mmcapacity}, which ensures a sequence of a uniformly spaced measurement samples.
% Throughput is measured chunk-by-chunk at APP, thus an irregularly sampled sequence (ISS) occurs as a result of chunk download time fluctuaions. This could significantly degrade the accuracy of throughput prediction. In contrast, measurements at PHY offer a more granular, slot-level throughput prediction, leading to more consistent throughput prediction. 
% \begin{figure}[htb]
%     \centering
%     \includegraphics[width = 0.5\textwidth]{photo/overview.pdf}
%     \caption{Architecture of the proposed cross-layer video delivery platform.}
%     \label{platform}
% \end{figure}
The selected bitrate at APP is that passed down to MAC in order to perform the video-aware resource allocation, aiming to maximize throughput utilization. Meanwhile, the allocated resource blocks (RB) count is transmitted to PHY as well, which is combined with a newly proposed soft link adaptation scheme there to decrease the error rate and improve the stability of throughput. The entire optimization process forms a closed loop, and the effectiveness of our proposed system was proved through extensive evaluations. 

Our contributions include
\begin{itemize}
    \item We construct the cross-layer video transmission platform, which encompasses the PHY, MAC, and APP layers. It can accurately emulate both the link transmission and playback of video delivery.
    \item We designed a three-level closed-loop cross-layer video delivery framework. This scheme effectively addresses the optimization needs of different layers. It also achieves cross-layer interaction and cooperative optimization.
    \item Extensive evaluations were conducted and we collected the distorted video streaming for further objective quality assessment to highlight the gains brought by our cross-layer design.
\end{itemize}

The rest of this paper is organized as follows. Sec. \uppercase\expandafter{\romannumeral 1} A. provides an overview of related work. Sec. \uppercase\expandafter{\romannumeral 2} formulates the cross-layer QoE maximization problem in consideration. In Sec. \uppercase\expandafter{\romannumeral 3}, we present the design of the cross-layer framework and Sec. \uppercase\expandafter{\romannumeral 4} gives the evaluation results. Finally, the paper is concluded in Sec. \uppercase\expandafter{\romannumeral 5}.

\subsection{Related work}
\textbf{Video delivery across different layers:}
Studies on video delvery are carried out in the PHY, MAC, transport, and APP layers of the communication network. Each layer employs distinct strategies and metrics for adaptive control. % while the network layer is preferred to be analyzed independently as its complex network topology is too difficult for simulation. 
In the PHY layer, link adaptation is the process of adaptively selecting modulation and channel coding rate (MCS) guided by channel feedback to enhance channel throughput and minimize transmission errors. The default 3GPP \cite{3gpp} method of link adaptation, known as inner loop link control (ILLA), employs a pre-established SNR-MCS lookup table to determine the choice of MCS based on real-time SNR estimates. To cope with SNR variations and delays, outer loop link adaptation (OLLA) is proposed by adding an extra offset to the SNR estimate, which is adjusted according to positive/negative acknowledgements (ACK/NACKs) received in the past transmissions \cite{OLLAdes}. More recently, \cite{Fan2011MCSSF, Peralta2022OuterLL, BayesianLA} employ either a rule-based or learning-based strategy for MCS adaptation based on multi-dimensional channel states, rather than just the SNR, to further boost the performance of link adaptation. However, the impact of link adaptation on throughput is localized as it only affects the size of each individual RB. \cite{MAXMINres} further requires resource allocation at MAC to make fine adjustments to the number of RBs. Existing resource allocation schemes are either based on proportional fairness or round-robin strategy \cite{globecom, RLresource}, which achieves equitable distribution by minimizing differences in resource allocated to different users. They largely ignore the video information from the APP layer.

% As we move up to the transport layer, methods like congestion control and forward error correction (FEC) are implemented to ensure reliable and smooth transmission. The Google congestion control (GCC) algorithm \cite{holmer2016google} is highly effective in controling the encoding rate, with the goal of preventing rebuffering and improving the viewing experience in real-time streaming. Since the transport layer is less relevant to our work, we shall not discuss it in greater detail here.
% \cite{fec} employs reinforcement learning to learn FEC control, which trades off the transmission efficiency and packet error rate. 

The APP layer has consistently been a focus in video delivery, primarily because it enables macro-level control over video quality and bitrate. Based on the architecture of traditional HTTP adaptive streaming, a collection of ABR approaches have been proposed to ameliorate user’s QoE. They can be classified as heuristic-based and learning-based. Heuristic-based ABRs \cite{festive, bola, mpc, oboe, cs2p} establish predefined rules that consider buffer status, estimated throughput, or both, to determine bitrate, providing a straightforward and effective solution. For instance, FESTIVE \cite{festive} predicts the future chunk throughput based on history measurements and selects the video birate that is closest to the predicted throughput. Since accurate throughput prediction is hard in reality, buffer-based ABRs \cite{bola, BBA} avoid direct throughput measurement by choosing the bitrate based on the playback buffer occupancy. To achieve a good trade-off between rate and buffer goals, MPC \cite{mpc} employs a hybrid strategy which first forecasts future chunk throughput using past data and solves the optimization problem for delivering upcoming chunks by integrating both bitrate and buffer considerations. On the other hand, learning-based ABRs \cite{pensieve, comyco, zwei} utilize advanced deep learning techniques, making the derived methods applicable to a wider variety of situations. Pensieve \cite{pensieve} employs deep reinforcement learning (DRL)-based framework for ABR by optimizing the QoE function defined in \cite{mpc}.

% Comyco \cite{comyco} employs the lifelong imitation learning to directly approximate the offline near-optimal expert solution. To learn ABR strategy better from the actual requirements, Zwei \cite{zwei} leverages the self-play reinforcement learning to obtain the video bitrate selection via adversarial competition.   

Although these studies have established remarkable control mechanisms, they lack integration of and interaction across different layers. This has motivated the study of optimization strategies that span multiple layers.

\textbf{Cross-layer designs:}
As the bitrate decisions at APP form the foundation for video delivery, most cross-layer optimization methods in video delivery involve the integration of APP with other lower layers \cite{haochen, MM1}.
% The coordination between the transport layer and the application layer is often considered the most prominent cross-layer issue, due to their close proximity and tight interconnection. 
% \cite{huanhuan, haochen} notice the separation between the transport and application layers could lead to a suboptimal user experience, primarily because of the misalignment in codec and transport. To tackle this issue, they utilize Deep Reinforcement Learning (DRL) to simultaneously optimize rate control in the transport layer and video coding in the application layer, resulting in improved performance in mobile video telephony. Joint optimization with the MAC layer adds the resource allocation to the bitrate constraints. 
% Similarly, we will bypass discussions on the cross-layer work between the APP and transport layer. 
% The main idea here is to integrate the video codec with rate control to avoid misalignment between the two in real-time video streaming \cite{huanhuan, haochen}.
In the following, we'll review cross-layer designs that merge APP with the wireless link. This consists of two aspects: integrating ABR at APP with MAC and PHY.

Current cross-layer work with MAC adds the constraint of limited resources within the adaptive streaming process, thereby achieving optimized resource allocation. For instance, \cite{RLresource} introduces an DRL-based resource allocation strategy to tradeoff QoE fairness among different users. \cite{elsevier} proposes a pseudo-polynomial time optimal algorithm to adaptively adjust video encoding configuration under given resource constraint. 
\cite{index} jointly optimizes the multi-user access and video bitrate under a large number of users associated with limited transmission time slots. Meanwhile, \cite{xr} deploys a multi-step Deep Q-network for frame-priority-based wireless resource scheduling and a transformer-based proximal policy optimization (TPPO) for video bitrate adaptation. \cite{appslice} proposes an adaptive GoP-level FEC allocation strategy based on the video slice header information. But it does not include the adaptation for video bitrate. Nonetheless, these MAC-APP cross-layer designs either consider bitrate adaptation under limited resources or re-design resource allocation given a fixed APP information. Given the complexity of channel variations, PHY is often referred to as "the last mile transmission", showing their significant impact on video delivery. Consequently, many studies have been devoted to optimizing the video delivery with the PHY layer information. \cite{pi} enables continuous monitoring of the LTE base station's cellular data to estimate the throughput, which is used to optimize throughput prediction in the APP layer.
% At the same time, the MPEG SAND-DASH architecture represents a significant advancement in enhancing cross-layer optimization by incorporating wireless network data into DASH \cite{sand}.
Similarly, \cite{crossbmsb, Prius, ofdma, access} considers video delivery on the physical downlink shared channel (PDSCH) and makes bitrate decisions based on link throughput measurement. \cite{xiong} explores adaptive video delivery for multiple users across a time-varying and mutually interfering multi-cell wireless network. They consider the impact of the video content complexity on resource allocation and conducts offline training of the ABR agent based on PHY network throughput traces. \cite{appslice} dynamically selects the MCS for video block transmission under given video source to maximize user experience. Cross-layer designs integrating PHY and APP face similar problems as those combining MAC. For instance, in \cite{xiong}, they first collect network traces under their proposed resource allocation scheme and then use the collected trace for offline training of the ABR agent. The lack of interaction between layers in existing cross-layer methods highlights the need for a more comprehensive evaluation platform for cross-layer designs in video delivery.
% \caption{Evaluation on classical ABRs on cross-layer delivery platform}

% \begin{tabular}{c c c c c}
% \hline
% ABRs & SSIM1 & BLER & BER & QoE \\
% \hline
% FESTIVE & 0.82 & 0.14 & 0.0056 & 4.15 \\
% BBA & 0.85 & 0.15 & 0.0063 & 3.50 \\
% BOLA & 0.84 & \textbf{0.17} & \textbf{0.0081} & \textbf{3.45} \\
% HYB & 0.8 & 0.16 & 0.0062 & 3.71 \\
% MPC & \textbf{0.77} & 0.14 & 0.0059 & 4.17 \\
% \hline
% \end{tabular}

% \end{table}

% \begin{table}
% \centering
% \caption{Evaluation on classical ABRs on cross-layer delivery platform}
% \resizebox{0.5\textwidth}{!}{
% \begin{tabular}{c c c c c c c c c }
% \hline
% ABRs & R (Mbps) & B (s) & Rebuf (s) & X & SSIM1 & SSIM\_all & BLER & BER & QoE \\
% \hline
% FESTIVE & \textbf{4.28} & \textbf{21} & 0.02 & 18 & 0.82 & 0.49 & 0.14 & 0.0056 & 4.15 \\

% BBA & 5.03 & 14 & 0.004 & 25 & 0.85 & 0.0.39 & 0.15 & 0.0063 & 3.50\\

% BOLA & 5.12 & 13 & 0.01 & \textbf{31} & 0.84 & 0.26 & \textbf{0.17} & \textbf{0.0081} & \textbf{3.45} \\

% HYB & 5.30 & 4 & 0.02 & \textbf{31} & 0.8 & \textbf{0.24} & 0.16 & 0.0062 & 3.71 \\

% RMPC & 5.61 & 2.5 & \textbf{0.05} & 20 & \textbf{0.77} & 0.42 & 0.14 & 0.0059 & 4.17\\

% \hline
% \end{tabular}}
% \label{nerv model}
% \end{table}

\section{Cross-layer QoE Maximization}

\subsection{Framework}
We consider the single-user 5G video delivery scenario shown in Fig. \ref{fig:str}, as the multi-user case inevitably involves a tradeoff between the fairness and each user's QoE. Later, we are going to show how our cross-layer design can be extended to the scenario with multiple users. 

The system includes a DASH server and a user equipment (UE). It has three layers, PHY, MAC and APP. At APP, the server stores several video files, each is divided into smaller chunks of $T_{\rm chunk}$ seconds. Without loss of generality, we assume there are one video file consists of $D$ chunks, indexed as $\{1,2,...,D\}$.
% At APP, the DASH server stores $D$ video files, denoted as $\mathcal{D} = \{1, 2, ..., D\}$, each of which is divided into smaller video chunks of $T_{\rm chunk}$ seconds. Thus, any video file $d \in \mathcal{D}$ can be viewed as a sequence of consecutive video chunks that are indexed by $\{1, 2, ..., D_d\}$. Here, $D_d$ denotes the number of chunks in video file $d$. 
Each video chunk is encoded into $L$ different representations (i.e., quality levels), with $a_{i, l}\ (i= 1,2,...,D; l\in\{1,2,...,L\})$ denoting the $l$-th bitrate for chunk $i$. The available representations for each video chunk are specified in the Media Presentation Description (MPD) file, which is periodically updated. DASH users can access the MPD file by sending an HTTP request. Upon requesting the download of a new chunk, the requested chunk is delivered over the 5G PDSCH. The throughput of this link is jointly affected by the allocation of RBs at MAC and the MCS at PHY. Prior to a new download request, UE needs to make the bitrate decision for the chunk to be delivered next, based on link throughput and playback information.
\begin{figure*}[!h]
    \centering
    \includegraphics[width = 0.99\textwidth]{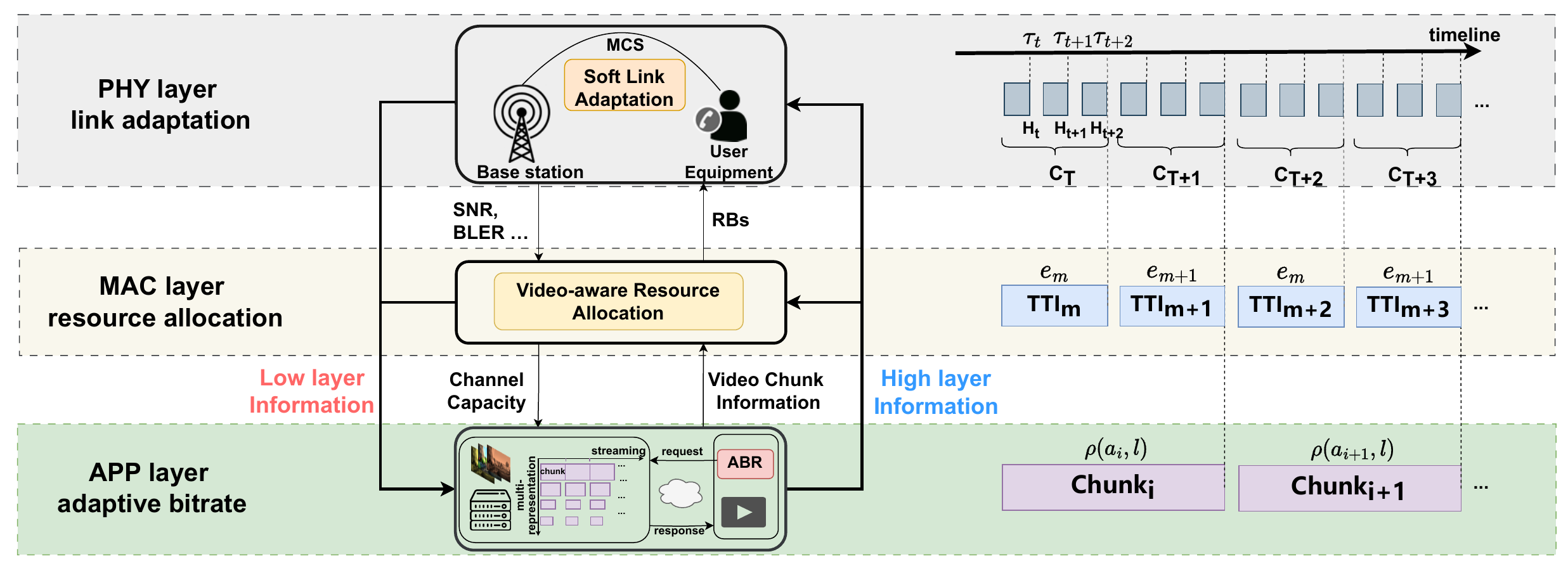}
    \caption{A schematic diagram of the cross-layer adaptive video delivery system.}
    \label{fig:str}
\end{figure*}

\subsection{Wireless Link}
We consider a single-link multipath channel with slow block fading. The channel impulse response remains stationary during one time slot \cite{eOLLA}. For time slot $t$, the user sends one transport block (TB) of several OFDM symbols, which are collected in the columns of $ \boldsymbol{\rm x}_t$. Assuming a time-varying but linear channel model with channel coefficient matrix $\boldsymbol{\rm H}_t$, normalized symbols and additive Gaussian noise $\boldsymbol{\rm n}_t$. we can write the received TB at the UE as
\begin{equation}
 \boldsymbol{\rm y}_t = \boldsymbol{\rm H}_t \boldsymbol{\rm x}_t + \boldsymbol{\rm n}_t.
 \label{1}
\end{equation}

With the 5G PDSCH, data streams from upper layers are encapsulated into separate TBs for "the last mile transmission". Each TB is composed of several RBs, and RB is the smallest allocation unit in the time-frequency domain. The scheduler in MAC determines the number of RBs at each transmission time interval (TTI), typically at intervals that are multiples of the time slots (i.e. 1 TTI = $m$ time slots). The number of allocated RBs $e_t$ at time slot $t$ must satisfy
\begin{equation}
     0 \leq e_t \leq e_{max}, 
     \label{RB}
\end{equation}
where $e_{max}$ is the total number of available RBs. Clearly, the number and size of allocated RBs directly affects the size of the TB. In particular, the size of each TB is determined by the MCS. MCS is selected by the base station at each time slot based on signal-to-noise ratio (SNR) estimation from a UE. As specified in \cite{3gpp}, every MCS is associated with a SNR range in a lookup table ($\mathcal{T} = \{\tau_1, \tau_2, ..., \tau_T\}$). % The base station frequently acquires SNR estimates via the wireless link. These estimates are subsequently utilized to refer to the appropriate SNR range in the MCS lookup table and select the corresponding MCS to transmit TBs in every time slot. 

We use the exponential effective SNR mapping (EESM) \cite{EESM} to calculate the SNR $\gamma_t$ of a TB as
\begin{equation}
    \gamma_t = \beta \ln\left(\frac{1}{N} \sum_{k=1}^{N} \exp\left(\frac{||\mathbf{\hat{h}}_k||^2/\left(y_{t,k} - \mathbf{\hat{h}}_{t,k}^T \mathbf{x}_t\right)^2}{\beta}\right)\right),
    % x_k \in \boldsymbol{\rm X}_t, y_k \in \boldsymbol{\rm Y}_t, \hat{h}_k \in \boldsymbol{\rm \hat{H}}_t,
\end{equation}
where $y_{t,k}$ is the $k$-th element in $\mathbf{y}_t$. $\mathbf{\hat{h}}^T_{t,k}$ is the $k$-th row of 
$\hat{\boldsymbol{\rm H}_t}$, which is the estimated channel state information (CSI) and can be acquired via channel estimation. $\beta$ is a tuning parameter. 

We employ a probabilistic formulation to model the link adaptation strategy. Specifically, let $p_{\tau,\gamma}$ be the joint density that indicates the probability of choosing MCS $\tau$ at SNR $\gamma$. In order to ensure reliable transmission, it is essential to impose the constraint
\begin{equation}
\sum_{\tau} \int_{\gamma} p_{\tau, \gamma} \eta_{\tau, \gamma} d\gamma \leq \eta_0.
    \label{error}
\end{equation}
The block error rate (BLER) $\eta_{\tau, \gamma}$ is the likelihood of incorrectly decoding TB with MCS $\tau$ at SNR $\gamma$. Eq. (\ref{error}) states that the average BLER of the system is expected to be below a predefined threshold $\eta_0$. Surpassing this threshold means an increase in the video packet error rate, causing decoding issues. The TB needs to be re-transmitted as long as its decoding is incorrect, thereby leading to increased transmission delay. The default link adaptation strategy is to directly select a MCS using the deterministic MCS lookup table, according to the SNR estimate. The wireless link also employs the hybrid automatic repeat request (HARQ) mechanism, which sends ACK or NACK signals to show whether each TB has been received successfully.
\begin{table*}[htbp]
\centering
\caption{Nomenclature.}
\begin{tabular}{|c|p{0.3\textwidth}||c|p{0.3\textwidth}|}
\hline
\textbf{Symbol} & \textbf{Description} & \textbf{Symbol} & \textbf{Description} \\ \hline
\multicolumn{4}{|c|}{\textbf{Video playback notations}}\\
\hline
$\mathcal{D} = \{1, 2, ..., D\}$ & The set of $D$ video streams. & $\{1, 2, ..., D\}$ & The set of video chunks in a video file. \\
$\phi_i$ & The rebuffering time for chunk $i$ & $T_{\rm chunk}$ & The duartion of a video chunk. \\

$C_i$ & The average throughput at chunk $i$. & $b_{i}$ & The buffer size at chunk $i$. \\

$t_i$ & The time to start downloading chunk $i$. & $d_i$ & The download time for chunk $i$. \\

$a_{i, l} \in \mathcal{A}_{D}$ & The bitrate of $i$-th chunk, $l$-th representation. & $\rho(a_{i, l})$ & The size of $i$-th chunk, $l$-th representation.\\
\hline
\multicolumn{4}{|c|}{\textbf{Link notations}}\\
\hline

$\boldsymbol{\rm X}_t = [1_t, 2_t, ..., x_t]$ & The sequence of transmitted symbols. & $\Phi$ & The threshold for discriminating soft-ACK. \\

$\boldsymbol{\rm Y}_t = [1_t, 2_t, ..., y_t]$ & The sequence of received symbols of $\boldsymbol{\rm X}_t$. & $\boldsymbol{\rm H}_t = [1_t, 2_t, ..., h_t]$ & The sequence of CSI. \\

$\mathcal{T} = \{\tau_1, \tau_2, ..., \tau_T\}$ & The set of all available MCSs defined in 3GPP. & $e_t$ & The allocated RB count at $t$. \\

$N_{\tau_t,e_t}$ & The size of TB assigned with $\tau_t$ and $e_t$. & $\gamma_t$ & The measured channel SNR. \\

$p_{\tau, \gamma}$ & The BLER under MCS $\tau$ and SNR $\gamma$. &    &  \\

$\eta_{\gamma_t, \tau_t}$ & The BLER under SNR $\gamma_t$ and MCS $\tau_t$. & $\Delta T_{\rm slot}$ & The length of a time slot in seconds. \\

$\delta_t$ & The OLLA offset at time $t$. & $\Delta_{\rm down}$, $\Delta_{\rm up}$  &  Outer loop link adaptation deviation parameters. \\
\hline
\end{tabular}
\label{tab:my_label}
\end{table*}
\subsection{Video playback}
We consider a situation where a user requests to watch a video file and downloads video chunks of the specified representations sequentially. Recall that $a_{i,l}$ is the bitrate chosen for the $i$-th chunk at $l$-th representation. Suppose at time $t_i$, the user initiates the downloading of video chunk $i$. We assume that the player will request the next chunk immediately after receiving the current one. That is, $t_{i+1} = t_{i} + d_{i}$, with $d_{i}$ denoting the downloading time for chunk $i$. Let $\rho(a_{i, l})$ be the chunk size under chunk representation \(a_{i, l}\). The downloading time \(d_i\) would be $d_{i} = \rho(a_{i,l}) / C_{i}$, where $C_{i}$ is the network throughput during the download process. $C_{i}$ can be found by computing the mean of the instantaneous transmission rates across all time slots while chunk $i$ is being downloaded as
\begin{equation}
    C_i = \frac{1}{t_{i+1}-t_{i}}\sum_{t = t_i}^{t_{i+1}}\frac{N_{\tau_t,e_t}}{\Delta  T_{\rm slot}} (1-\eta_{\tau_t, \gamma_t}).
\label{6}
\end{equation}
$N_{\tau_t,e_t}$ is the TB size under MCS $\tau_t$ and RB number $e_t$. $\Delta T_{\rm slot}$ is the duration for each time slot. The rightmost term of Eq. (\ref{6}) is actually the inverse of the average number of transmissions required for a successful chunk decoding. In practice, we do not know the exact value of $\eta_{\tau_t, \gamma_t}$ and it is replaced by the proportion of retransmissions during the chunk delivery. The network throughput $C_i$ is upper-bounded by the Shannon's capacity \cite{shannon2001mathematical} as
\begin{equation}
    C_i \leq \frac{1}{t_{i+1}-t_{i}}\sum_{t = t_i}^{t_{i+1}}e_t \cdot \text{log}_2\left (1+\frac{\gamma_t}{e_t}\right ).
    \label{7}
\end{equation}
% Previous cross-layer designs involving PHY merely use the fixed throughput measured at wireless link or directly treat link capacity as throughput \cite{xiong, crossbmsb}, overlooking the impact of resource allocation on throughput.

The received video chunks are stored in the user playback buffer. The buffer occupancy when the user starts downloading chunk \(i\) is denoted as \(b_{i} \in [0, b_{\text{max}}]\). Rebuffering would occur when the buffer is depleted, indicating that the download time \(d_{i}\) exceeds the buffer level \(b_{i}\). Rebuffering time $\phi_i$ can thus be expressed as
\begin{equation}
 \phi_i = (d_{i} - b_{i})_{+} = \left(\frac{\rho(a_{i, l})}{C_{i}}-b_{i}\right )_{+},
 \label{9}
\end{equation}
where \((x)_+ = \max\{x, 0\}\). We can write the evolution of the playback buffer occupancy as
\begin{equation}
    b_{i+1} = (b_{i} - d_{i}) + T_{\text{\rm chunk}}.
    \label{10}
\end{equation}

\subsection{Problem Formulation}
To ensure a smooth playback, frequent quality variations and rebuffering events should be minimized while maximizing the bitrates of transmitted videos. We define the reward for measuring user QoE \cite{mpc} as 
\begin{equation}
\label{QoE}
\text{QoE}_{i} = a_{i,l} - \alpha |a_{i,l} - a_{i-1,l}| -  \beta \phi_{i}, 
\end{equation}
where $\alpha$ and $\beta$ are non-negative parameters used to adjust the importance of quality ﬂuctuations and rebuffering events in the QoE evaluation.

With the user mobility and channel fading, the wireless channel conditions between base stations and mobile users vary over time. In contrast, watching a video clip typically takes a relatively longer time. To ensure a high QoE for the user over this extended period, we formulate the overall QoE optimization for video delivery as
\begin{equation*}
    \begin{aligned}
        &\rm \text{\textbf{P1:}} &&\underset{\boldsymbol{\tau},\boldsymbol{e},\boldsymbol{a}}{\rm \text{max}}
         \sum^{D}_{i = 1}\text{QoE}_{i}\\
        &\rm \text{\textbf{s.t.}}
        & & \text{Constraints from Eqs.}\  (1), (2), (4), (6), (7)\text{, and } (8). \\
        % & & & 0 \leq C_{\it{i+1}} \leq C^{\prime}_{\it{i}}, \\
    \end{aligned}
\end{equation*}
Decision variables for this problem are 1) MCSs at PHY: $\boldsymbol{\tau} = \{\tau_t\}_{t = 1,2,...}$, adjusted every time slot; 2) number of RBs at MAC $\boldsymbol{e} = \{e_t\}_{t = m,2m,...}$, adjusted every TTI and 3) video bitrate at APP $\boldsymbol{a} = \{a_{i,l}\}_{i = 1,2,...,D;\ l \in \{ 1, ... , L\}}$, adjusted every chunk. This is a NP-hard problem and involves three decision spaces with varying timescales. To simplify the above cross-layer optimization problem, we decompose the it into three parts, namely (PHY) link adaptation, 
(MAC) resource allocation and (APP) bitrate adaptation.

\section{Design of StreamOptix}
\subsection{System Overview}
With the proposed StreamOptix system, at the PHY layer, channel capacities are measured uniformly in time to produce the regularly sampled sequence (RSS). The RSS is input to APP to determine the chunk bitrate using MPC. The optimal bitrate that maximizes the future chunks' QoE is found. The selected bitrate is then sent to lower layers to guide the wireless link control. Specifically, MAC allocates resource to maximize the rate utilization given the target video bitrate. The selected RB number is forwarded to PHY to constrain the link adaptation at smaller timescales and minimize the transmission errors. During the link adaptation process, PHY keeps gathering CSI to obtain link capacity measurements, which are sent to APP for future bitrate decisions. Our design thus leads to the establishment of a closed-loop control system whose modules (i.e., layers) operate at different timescales.

\subsection{PHY: Soft Link Adaptation}
According to Eq. (\ref{6}), the system throughput is limited by the link BLER in the sense that a large BLER could cause an increased number of retransmissions. In this case, transmitting high-bitrate videos may not improve the user QoE, as it could result in frequent rebuffering (see Eq. (\ref{9}) and (Eq. \ref{QoE})). 

To improve the quality of video delivery and achieve more robust link performance, an efficient link adaptation scheme is developed to cope with changing network conditions. In contrast to the chunk-level decisions made at APP, link adaptation aims at inducing macroscopic changes through small-scale adjustments. Hence, when designing link adaptation schemes, it is essential to take into account our targets (i.e. throughput or BLER) in terms of their average values. 
% This has perfectly match the rate adaptation mechanism in APP, where decision intervals are usually sufficiently large to be considered as "mathematical expectations." 

In order to maintain a low BLER while transmitting high-quality video, we formulate the link adaptation problem as
\begin{equation*}
    \begin{aligned}
        &\rm \text{\textbf{P2:}} &&\underset{\{\tau_t\}_{t = 1,2,...}}{\rm \text{max}}
         \sum_{t}N_{\tau_t, e_t}(1-\eta_{\gamma_t, e_t})\\
        &\rm \text{\textbf{s.t.}}
        & & \text{Constraints from Eqs. }  (1), (3), (4),\ \text{and}\ (6). \\
        % & & & 0 \leq C_{\it{i+1}} \leq C^{\prime}_{\it{i}}, \\
    \end{aligned}
\end{equation*}
The default 3GPP method of link adaptation employs the MCS lookup table to determine MCS based on SNR estimates. The MCS lookup table \cite{3gpp}, generated by 3GPP through offline simulations under various link conditions, identifies the MCS that maximizes throughput at different SNR levels. Nevertheless, the actual link often exhibits complex variations that prevent achieving the desired BLER performance using this table alone. To address this issue, OLLA was designed to dynamically incorporate an offset related to HARQ feedback into the estimated SNR, aiming to mitigate excessive BLER. While this method does lower BLER, it exhibits slow convergence. Therefore, we integrate a novel mechanism called soft-ACK with OLLA to further enhances its BLER convergence and efficacy.

\textbf{OLLA:} In the default link adaptation mode specified by 3GPP, the base station (BS) computes the SNR based on the CSI estimation at each time slot and subsequently chooses the MCS according to a predefined lookup table. However, the MCS decisions could be prone to inaccuracies due to channel fluctuations and CSI reporting delays. These lead to suboptimal performance and may not ensure fast BLER convergence. In order to address this weakness, OLLA is proposed to handle SNR inaccuracies by modifying the measured SNR values from the CSI report using an estimated offset $\delta_{t}$ before linking them to the MCS index. % The OLLA offset is fine-tuned based the HARQ feedback using to consistently maintain peak system performance and guarantee a top-norch user experience in a changing wireless environment. 
Specifically, OLLA fine-tunes the SNR estimate using
\begin{equation}
    \hat{\gamma}_t = \gamma_{t} - \delta_{t}.
\end{equation}
Based on the HARQ feedback, the offset $\delta_{t}$ is updated as
\begin{equation}
\delta_t = \begin{cases} 
\delta_{t-1} + \Delta_\text{up}, & \text{if NACK} \\
\delta_{t-1} - \Delta_\text{down}, & \text{if ACK}.
\end{cases}
\end{equation}
Each ($\Delta_\text{up}, \Delta_\text{down}$) pair corresponds to different BLER convergences and is set according to the specific requirement and link conditions. When the UE HARQ feedback is NACK, the offset is incremented; conversely, it is decremented. 

In the following, the convergence of the OLLA algorithm under stationary link conditions, as well as the conditions that $\Delta_\text{down}$ and $\Delta_\text{up}$ should be satisfied in order to achieve convergence, is studied.
\begin{Theorem}
The BLER of OLLA will converge to $\frac{1}{1+\frac{\Delta_\text{\rm up}}{\Delta_\text{\rm down}}}$ if $\Delta_\text{\rm down} + \Delta_\text{\rm up} < \frac{2e}{1.11}$.
\end{Theorem}
\begin{proof}
Let $\rm I(.)$ denote the indicator function. We can rewrite (10) and (11) as
\begin{equation}
    \delta_t = \delta_{t-1} + \rm{I(NACK)} \Delta_\text{up} - \rm{I(ACK)} \Delta_\text{down}.
    \label{ollains}
\end{equation}
Taking expectation on both sides of (\ref{ollains}) yields the recurrence function for $\delta$, which is
\begin{equation}
   \label{reccurence equation}
   T(\delta) = \delta + \eta_t \Delta_\text{up} - (1- \eta_t) \Delta_\text{down},
\end{equation}
where $\eta_t$ is the average BLER at time $t$. 

According to \cite{eOLLA}, for a random sequence to converge, its recurrence function must satisfy both the first and second Banach fixed point conditions. 
The first Banach fixed point condition states that the recurrence function should be a contract mapping. That is, we require that
\begin{equation}
\begin{aligned}
T(\delta) \in \left[\delta_{\text {min }}, \delta_{\text {max }}\right], \forall \delta \in \left[\delta_{\text {min }}, \delta_{\text {max }}\right].
\end{aligned}
\end{equation}
Firstly, the values of $\eta_t$ are between 0 and 1. With the continuing reception of ACKs, there exists a sufficiently small $\delta_{t} = \delta_{\text{min}}$ such that an overestimate of the SNR is obtained, thereby yielding eventually $\eta_t = 1$ and an NACK being received. In this case, we have $T\left(\delta_{\min }\right)=\delta_{\min }+\Delta_{\rm up}$. Next, as $\delta_{t}$ increases, the value of $\eta_t$ starts to decrease until $\delta_{t}$ reaches a value $\delta_{\max }$ high enough to ensure that the SNR is underestimated and $\eta_t = 0$. In this case, we obtain $T\left(\delta_{\max }\right)=\delta_{\max }-\Delta_{\text {down }}$. In summary, it can be concluded that
\begin{equation}
\begin{aligned}
T\left(\delta \right) &\in \left[\delta_{\text {min }}+\Delta_{\rm u p}, \delta_{\text {max }}-\Delta_{\text {down }}\right] \\
&\in \left[\delta_{\text {min }}, \delta_{\text {max }}\right], \forall \delta \in \left[\delta_{\text {min }}, \delta_{\text {max }}\right].
\end{aligned}
\end{equation}
The second Banach fixed point condition states that 
\begin{equation}
 \left|T^{\prime}\left(\delta \right)\right|<1, \forall \delta \in\left[\delta_{\min }, \delta_{\max }\right].
\end{equation}
The proof for the second Banach fixed point condition is included in Appendix A. 

As a result, according to the Banach fixed-point theorem \cite{banach}, when $t$ approaches infinity, the recurrence equation in (\ref{reccurence equation}) converges to the unique solution, which is obtained through setting $T(\delta)=\delta$, 
\begin{equation}
\label{17}
    \underset{t\rightarrow \infty}{\text{lim}}\eta_t = \frac{1}{1+\frac{\Delta_\text{up}}{\Delta_\text{down}}}.
\end{equation}
\end{proof}

By accumulating the small adjustments triggered by ACKs and NACKs over time, OLLA eliminates the discrepancy between the target BLER and the estimated BLER. The offset $\delta_{t}$ will finally reach a steady state when the system converges to the BLER target. 

\textbf{Soft Link Adaptation:} Empirically, $\Delta_\text{down}$ is set to be a small value for OLLA to achieve the BLER target. However, this could result in a lengthy convergence period much longer than the duration of a chunk. Moreover, depending solely on OLLA for offset adjustments may not enable quick discovery of poor channel conditions, since OLLA reduces the SNR estimate only after receiving a NACK. Meanwhile, the network conditions may have deteriorated significantly. Therefore, it is crucial to develop an approach that achieves faster convergence and lower BLER. For this purpose, we redefine the ACK levels, making them dependent on the channel's metrics (i.e., BLER). This is referred to as soft-ACK. Soft-ACK divides the prediction intervals of a specific channel metric to categorize various ACK levels \cite{Nemeth2021SoftACKFB}. It captures more channel information by increasing the ACK granularity, thus attaining higher sensitivity to channel degradation. % We shall utilize two tiers of soft-ACK based on BLER estimation.

% \begin{algorithm}
%   \caption{Soft Link Adaptation Algorithm}
%   \label{alg1}
%   \begin{algorithmic}
%   \STATE \textbf{Initialize:} $\rm MCS-Table: \gamma \rightarrow \tau_{1,2,...15}$; Soft-ACK threshold $\Phi$
%   \STATE \textbf{Initialize:} $\delta_t, \Delta_\text{down}, \Delta_\text{up}$
%   \STATE \textbf{Input:} Estimated SNR $\gamma_t$ from CSI; HARQ feedback flag from UE
%   \FOR{each transmission $t$}
%   \STATE $\gamma_t = \gamma_t - \delta_t$
%   \STATE select $\tau_t$ from  $\gamma \rightarrow \tau_{1,2,...15}$
%   \IF{flag = ACK}
%   \STATE estimate the block error probability $\hat{\eta_t}$
%       \IF{$\hat{\eta_t} \leq \Phi$}
%       flag = HIGHACK
%       \ELSE
%       flag = LOWACK
%       \ENDIF
%   \ENDIF
%   \IF{flag = HIGHACK}
%   \STATE $\delta_t \leftarrow \delta_t - \Delta_\text{down}$
%   \ELSIF{flag = LOWACK, NACK}
%   \STATE $\delta_t \leftarrow \delta_t + \Delta_\text{up}$
%   \ENDIF
%   \ENDFOR
%   \end{algorithmic}
% \end{algorithm}

We introduce two levels of ACKs: high-margin-ACK and low-margin-ACK. High-margin-ACK is treated as a "good" ACK which triggers a decrease by $\Delta_{\rm down}$ to the SNR offset $\delta_t$. On the other hand, a low-margin-ACK is considered as a NACK, leading to an increase in the offset by $\Delta_{\rm up}$. The distinction between high- and low-argin ACKs is made using the BLER threshold $\Phi$ \cite{Nemeth2021SoftACK}. If the BLER is below this threshold $\Phi$, it indicates an effective ACK (High-margin-ACK). If not, the ACK is categorized as low margin, indicating a potential channel quality degradation. Given that the BLER for each TB is not directly available, we utilize an empirical equation to approximate the BLER $\hat{\eta}_{\tau_t, \gamma_t}$ at $t$. This approximation is reliant on the estimated SNR and the MCS, which is \cite{Jenqneng}
\begin{equation}
\hat{\eta}_{\gamma_t,\tau_t} = 1 - \frac{1}{2}\left (1+\text{erf}\left(\frac{\gamma_{t}-\gamma[\tau_t]+\alpha}{\sqrt{2}}\right)\right).
\label{blerest}
\end{equation}
Here, $\gamma_t$ is the estimated SNR at $t$, $\tau$ is the selected MCS. $\gamma[\tau_t]$ is the SNR threshold of MCS $\tau_t$ defined in \cite{3gpp} and $\alpha$ is a tunable parameter. 

The convergence of soft-ACK under stationary link conditions, as well as the conditions that should be satisfiedm is studied in the following.
\begin{Theorem}
The BLER of soft link adaptation will converge to a target that is smaller than $\frac{1}{1+\frac{\Delta_\text{\rm up}}{\Delta_\text{\rm down}}}$ if the following conditions are hold:
\begin{itemize}
    \item The threshold $\Phi$ for determining low-margin-ACKs is greater than or equal to $\frac{1}{1+\frac{\Delta_\text{\rm down}}{\Delta_\text{\rm up}}}$.
    \item $\Delta_\text{\rm up} < \frac{2e}{1.11}$. 
\end{itemize}
\end{Theorem} 
\begin{proof}
We can write the transition for $\delta_t$ under soft link adaptation as 
\begin{align}
\delta_t = \delta_{t-1} &+ \mathrm{I}(\text{NACK}) \Delta_{\mathrm{up}} - \mathrm{I}(\text{ACK}) \Pr(\hat{\eta} \leq \Phi) \Delta_{\mathrm{down}} \nonumber \\
&+ \mathrm{I}(\text{ACK}) \Pr(\hat{\eta} > \Phi) \Delta_{\mathrm{up}}.
\end{align}
$\Pr(\hat{\eta}\leq \Phi)$ represents the probability that the estimated BLER, given in (\ref{blerest}), falls below the threshold $\Phi$. Under stationary channel conditions, $\Pr(\hat{\eta} \leq \Phi)$ can be reasonably approximated by the proportion of high-margin-ACKs $\eta(\Phi)$. The recurrence function for $\delta_t$ can be expressed as
\begin{equation}
\begin{aligned}
T(\delta) = \delta &+ \eta_t\Delta_{\rm up}
- (1-\eta_t)\eta(\Phi)\Delta_{\rm down} \\
&+ (1-\eta_t)(1-\eta(\Phi))\Delta_{\rm up}.
\label{soft dynamic}
\end{aligned}
\end{equation}
As in the proof of \textbf{Theorem 1}, the first and second Banach fixed point conditions should be satisfied to guarantee the convergence of BLER.   
As \(\delta_t\) becomes increasingly smaller, the SNR estimates rises until the BLER \(\eta_t\) reaches 1. This would lead to the same lower bound discussed in \textbf{Theorem 1}. On the other hand, as \(\delta_t\) increases and the value of \(\eta_t\) gradually drops to 0, the upper bound under soft link adaptation becomes
\begin{equation}
\delta_{\text{max}} - \eta(\Phi)\Delta_{\text{down}} + (1-\eta(\Phi))\Delta_{\text{up}}.
\end{equation}
In order to ensure that $T\left(\delta \right)$ is a contract mapping, we need to ensure the upper bound is smaller than $\delta_\text{\rm max}$, that is
\begin{equation}
\begin{aligned}
\delta_{\text{max}} - \eta(\Phi)\Delta_{\text{down}} + (1-\eta(\Phi))\Delta_{\text{up}} \leq \delta_{\text{max}} \\
\Rightarrow 1 > \eta(\Phi) \geq \frac{1}{1+\frac{\Delta_\text{\rm down}}{\Delta_\text{\rm up}}}.
\end{aligned}
\end{equation}

Proof for the second Banach fixed condition can be found in Appendix B.
When $t$ approaches infinity, we have that
\begin{equation}
    \underset{t\rightarrow\infty}{\text{lim}}\eta_t = \frac{(1 - \frac{1}{\eta(\Phi)})\frac{\Delta_\text{up}}{\Delta_\text{down}}+1}{1+\frac{\Delta_{up}}{\Delta_\text{down}}} < \frac{1}{1 + \frac{\Delta_\text{\rm up}}{\Delta_\text{\rm down}}}.
    \label{small}
\end{equation}
\end{proof}
Compared with (\ref{17}), the achieved BLER after convergence will be lower than that of OLLA. This analytically justifies the better performance of the proposed soft link adaptation scheme. By introducing BLER estimation, we utilized more channel information beyond raw ACK/NACK indicators for adjusting the channel SNR estimates. Furthermore, this method is simpler and exhibits stronger generalization capabilities than other link adaptation methods based on Bayesian inference or reinforcement learning (RL).
% \begin{algorithm}
%   \caption{Video-aware Resource Allocation}
%   \label{alg1}
%   \begin{algorithmic}
%   \STATE \textbf{Initialize:} Chunk-Duration $T_{\rm chunk}$; Resource Grid $\{0,1,...,e_{max}\}$; Multi-user Resource Allocation Strategy; Number of users M 
%   \STATE \textbf{Input:}  Chunk-Size: $\rho(a_i,l)$ from CSI
%   \FOR{each TTI}
%   \STATE do multi-user resource allocation to determine the allocated ratio $\epsilon0$ for each user 
%   \FOR{each user}
%   \STATE estimate the block error probability $\hat{\eta_t}$
%   \STATE calculate the TB size $N_0$ that just allows the current chunk to be transmitted without any rebuffering 
%   \STATE select $e_t$ that corresponds to the nearest TB size to $N_0$
%   \STATE calculate the allocated ratio as $\epsilon_1 = e_t/e_{max}$
%   \ENDFOR
%   \STATE multiply $\epsilon_0$ with $\epsilon_1$, and normalized over all users to obtain the final resource allocation ratio $\epsilon$ for each user
%   \ENDFOR
%   \end{algorithmic}
% \end{algorithm}

\subsection{MAC: Video-aware Resource Allocation}
The MAC layer efficiently employs a scheduler to manage the number of RBs allocated. As previously discussed, both the MCS and RB allocation influence the TB size, thereby affecting the link's throughput. Nevertheless, the MCS is determined by link adaptation, which must meet the BLER constraint. Adopting a more 'aggressive' MCS selection strategy will increase the block size, but it also enlarges the probability of erroneous block decoding, which in turn reduces the overall throughput. In contrast, resource allocation at the MAC layer affects the throughput with a minimal impact on BLER. Thus, optimizing resource allocation at the MAC layer can achieve better bandwidth control. During video streaming, the bitrate ladder shows significant gaps between consecutive blocks. These discontinuities coupled with imprecise throughput prediction result in a severe mismatch between the chunk bitrate and the actual throughput. Current resource allocation does not take into account the impact of throughput mismatches. Based on these considerations, we propose the following video-aware resource allocation that aligns the throughput with the chunk bitrate by acquiring the bitrate information from APP.

As the chunk length (with an order of seconds) is much longer compared to TTI (with an order of milliseconds) for resource allocation, it is challenging for MAC to quickly affect the bitrate decisions at APP. Therefore, we adopt a top-down optimization approach in which resource allocation at the MAC layer during the downloading of a chunk is constrained and guided by chunk bitrate selected by the APP layer. This aims to more closely align throughput the chunk bitrate to maximize rate utilization. Assuming $e_{max}$ resource blocks are available at the base station, we formulate an integer linear programming problem for video-aware resource allocation (VRA), which is
\begin{equation*}
    \begin{aligned}
        &\rm \text{\textbf{P3:}} &&\underset{\{e_t\}_{t = m,2m,...,d_i}}{\rm \text{\textbf{min}}}
         \sum^{D}_{i=1}|\sum^{d_i}_{t = 1}N_{\tau_t, e_t}(1-\eta_{e_t, \gamma_t}) - \rho(a_{i,l})|\\
        &\rm \text{\textbf{s.t.}}
        & & \text{Constraints from Eqs. } (1), (2), (3), (5),\text{and}\ (6). \\
        % & & & 0 \leq C_{\it{i+1}} \leq C^{\prime}_{\it{i}}, \\
    \end{aligned}
\end{equation*}
At time $t$ when decisions are made (at the start of each TTI), since the BLER $\eta_{e_t, \gamma_t}$ is unknown, we utilizes (\ref{blerest}) to estimate BLER $\hat{\eta}_{\gamma_t, \tau_t}$ from past observations. The TB size $N_0$ that best matches the bitrate of the chunk is calculated as
\begin{equation}
    N_{0} = \frac{\rho(a_{i,l})}{(1-\hat{\eta}_{\gamma_t, \tau_t})T_{\rm chunk}}.
\end{equation}
However, the number of RBs is an integer, while $N_0$ takes a positive real value. For practical implementation, we select $e_t$ from the set $\{0,1,...,e_{max}\}$, based on the known MCS $\tau_t$ from soft link adaptation, that generates the TB size $N_{\tau_t, e_t}$ closest to $N_0$. The simplicity of our approach lies in its algorithmic complexity being $O(n)$. While the estimation of $N_0$ may be rough, as the TTI of resource allocation increases during a single chunk download period, the total throughput tends to approach the bitrate of the downloaded chunk.
\subsection{APP: Adaptive Bitrate}
In APP, it is essential to make decisions regarding the chunk bitrate of a video file to optimize the overall QoE:
\begin{equation*}
    \begin{aligned}
        &\rm \text{\textbf{P4:}} &&\underset{\{a_{i,l}\}_{i = 1,2,...,D}}{ \text{max}}
         \sum^{D}_{i=1} \rm{QoE}_{\textit{i}}\\
        &\rm \text{\textbf{s.t.}}
        & & \text{Constraints from Eqs. }  (1), (5), (6),\ \text{and},\ (8).\\
        % & & & 0 \leq C_{\it{i+1}} \leq C^{\prime}_{\it{i}}, \\
    \end{aligned}
\end{equation*}

We employ the classic ABR strategy, MPC, as our bitrate decision module. Given that MPC employs a rate-based mechanism, it is able to fully leverage the benefits brought by our precise throughput prediction mechanism. Additionally, it combines the forecasted throughput in order to enhance the QoE for several upcoming chunks, demonstrating some extent of resilience. Nonetheless, this does not mean that our platform is not compatible with other ABR algorithms. We will demonstrate at the end of our paper that our cross-layer design brings various levels of enhancement to all current ABRs.\par

MPC gathers historical chunk throughput to forecast the throughput available for downloading several future chunks. It then conducts playback simulation based on the predicted throughput, selecting the first bitrate from the combination with the highest QoE scores as the decision bitrate for the forthcoming chunk. With the the abundance of timeslots in PHY transmission, the average channel capacities are calculated by measuring them over a set number of timeslots to creating a historical sequence that is uniformly distributed in time (e.g., every 500 slots). These measurements are used with a sliding window to compute the harmonic mean to predict the capacity of the upcoming channel condition. As mentioned previously, a significant advantage of utilizing measurements from PHY is the ability to create uniformly spaced throughput sequences (Fig. \ref{fig:mpctrace}), thereby removing the irregular time distribution caused by different download durations at chunk level measurement. 
%Details about the bitrate adaptation algorithm is shown in \textbf{Algorithm 2}.
\begin{figure}
    \centering
    \includegraphics[width=0.95\linewidth]{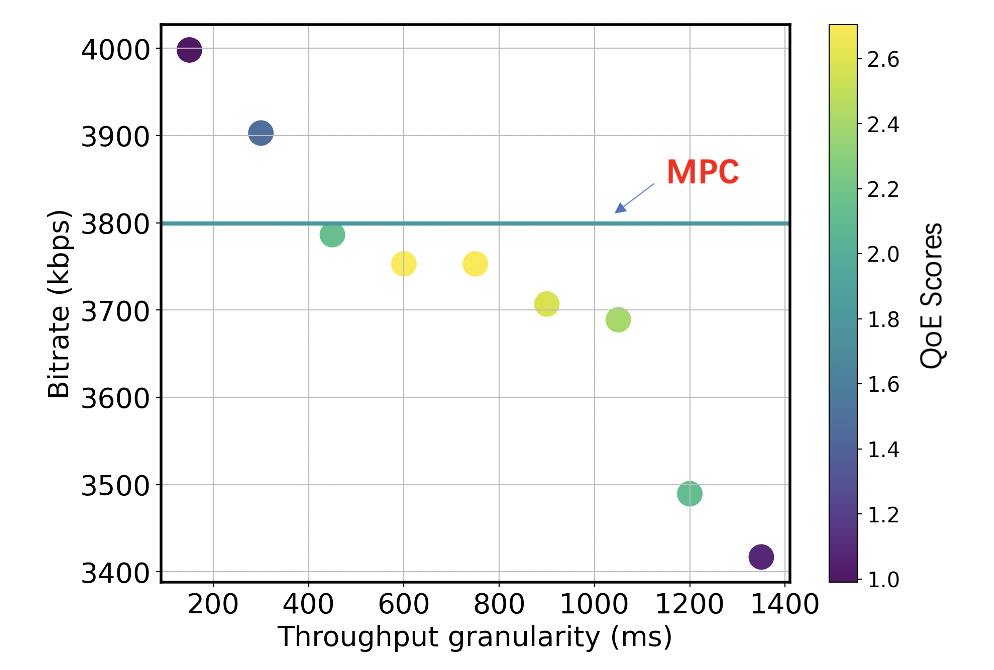}
    \caption{Uniformly spaced throughput in StreamOptix (upper) v.s. non-uniformly spaced throughput in chunk-level meaurement (bottom).}
    \label{fig:mpctrace}
\end{figure}

% \begin{algorithm}
%   \caption{Video adaptation flow on MPC-S}
%   \label{alg1}
%   \begin{algorithmic}
%   \STATE \textbf{Initialize:} Prediction window length $M$; Throughput prediction interval $T$; Video file length $D$; History window length $N$; 
%   \FOR {k = 1 to D}
%   \STATE $\hat{C}_{[t_k,...,t_k+M*T]}$ = ThroughputPred($C_{[t_k-N*T,...,t_k]}$)
%   \STATE $\text{QoE}_{max}=0$ 
%   \FOR {each possible M-length bitrate trace $a_{[k,...,k+M-1]}$}
%   \STATE $\text{QoE}_{[k,...,k+M-1]}=0$
%   \FOR{i = 1 to M}
%   \STATE $d_i = 0$; $\text{acc\_size} = 0$
%   \STATE \textbf{Calculate the transmission delay $d_i$:}
%   \FOR{j = 0}
%    \IF{$\text{acc\_size}+\hat{C}_{t_k+j*T}*T \geq \rho(a_i)$}
%    \STATE $d_i+=\frac{\rho(a_i)-\text{acc\_size}}{\hat{C}_{t_k+j*T}}$
%    \STATE \textbf{Break}
%    \ENDIF
%    \STATE $d_i+= T$; $\text{acc\_size}+=\hat{C}_{t_k+j*T}*T$; j+=1
%   \ENDFOR
%   \STATE \textbf{Calculate rebuffering:} $\phi_i = (d_i - b_i)_{+}$
%   \STATE \textbf{Calculate QoE:} 
%   $\text{QoE}_{i} = a_{i,l} - \alpha |a_{i,l} - a_{i-1,l}| -  \beta \phi_{i}$ 
%   \STATE \textbf{Update timeline:} $t_{i+1} = t_i + d_i$
%   \STATE \textbf{Update buffer occupancy:} $b_{i+1} = (b_i - d_i) + T_{\rm chunk}$
%   \ENDFOR
%   \STATE $\text{QoE}_{[k,...,k+M-1]}+=\text{QoE}_i$
%   \IF{$\text{QoE}_{[k,...,k+M-1]} \geq \text{QoE}_{max}$}
%   \STATE $\text{QoE}_{max} = \text{QoE}_{[k,...,k+M-1]}$
%   \ENDIF
%   \ENDFOR
%   \STATE Output the first element $a_{k+1}$ from the bitrate trace that achieves the maximum $\text{QoE}_{[k,...,k+M-1]}$
%   \ENDFOR
%   \end{algorithmic}
% \end{algorithm}

\subsection{Summary}
This section presents our proposed soft link adaptation at PHY, video-aware resource allocation at MAC, and MPC-S (MPC on StreamOptix) at APP. These mechanisms form a closed loop and reinforce each other. Unfortunately, this work does not extend to multi-user systems, primarily due to the incompatibility of \textit{MATLAB Simulink}'s custom link-level and system-level simulations. Moreover, multi-user scenarios introduce conflicts between individual optimality and fairness, making it challenging to determine the optimal solution. For completeness, we briefly discuss how StreamOptix can be extended to multi-user cases. For link adaptation, the adjustment needed is to replace SNR estimates with SINR (signal-to-interference and noise ratio) to account for the interference from other users. MAC is a little bit more complicated as it involves fairness among multiple users. 
Therefore, a two-layer resource allocation strategy can be considered, where the first layer employs the current multi-user resource allocation strategy, like round robin and proportional fairness, while the second layer adopts video-aware resource allocation. The ultimate resource distribution results from multiplying the ratios from both layers element-wise, using trade-off weights. APP does not require any further modification, as it makes bitrate decisions for each user separately.

\section{StreamOptix Implementation}
In this section, we describe the implementation of StreamOptix, as well as the extensive evaluation of our proposed cross-layer scheme over numerous metrics. To create StreamOptix, we use \textit{MATLAB} and \textit{Python} to construct the wireless link and video player respectively. 
% The reasoning behind this selection is two-fold. Firstly, since most ABRs currently run on \textit{Python} envirnment, using \textit{Python} to construct a video playback simulator facilitates the future integration of more ABR algorithms on this platform. Furthermore, \textit{MATLAB} offers robust link simulation tools for performing stable PHY simulations. Specifically, StreamOptix manages link control and resource allocation on \textit{MATLAB R2023}, while the integration of physical layer data and the video playback logic is handled in \textit{Python 3.8}. 
We use \textit{matlab.engine} API to enable real-time communication between wireless link deployed in \textit{MATLAB}'s and the virtual player in \textit{Python}. This setup guarantees the separation of higher and lower layer functionalities, enhancing scalability of our platform and enabling future extensions.
\begin{table}
\caption{Parameter configuration for wireless link.}
\centering
\resizebox{0.42\textwidth}{!}{
\begin{tabular}{|c|c|c|} 
\hline
\multicolumn{2}{|c|}{Scenario and Configuration Parameter} & Value \\ 
\hline
\multicolumn{2}{|c|}{Channel Type} & TDL C\\
\multicolumn{2}{|c|}{Subcarrier Spacing} & 15 kHz\\ 
\multicolumn{2}{|c|}{Channel Bandwidth} & 20 MHz\\
\multicolumn{2}{|c|}{Resource Grid} & 0-52\\
\multicolumn{2}{|c|}{Maximum HARQ Retransmissions} & 3\\
\multicolumn{2}{|c|}{LDPC Maximum Decoding Iterations} & 10\\
\multicolumn{2}{|c|}{Transport Block Duration} & 1 ms\\
\multicolumn{2}{|c|}{Target BLER} & 0.1\\
\multicolumn{2}{|c|}{Step size $\Delta_{\rm up}$} & 0.4 dB\\
\multicolumn{2}{|c|}{Maximum deviation} & 3 dB\\
\multicolumn{2}{|c|}{Soft-ACK Threshold $\Phi$} & 0.92\\
\multicolumn{2}{|c|}{Throughput Prediction Interval} & 600ms\\
\hline
Channel Conditions & SNR & Doppler Shift \\
\hline
Static & -5 - 5 dB & 10 Hz \\
Dynamic & -5 - 5 dB & 50 Hz \\
High SNR & 5 - 15 dB & 10 Hz \\
\hline
\end{tabular}}
\label{tab:combined}
\end{table}

\begin{table*}[h!]
\centering
\captionsetup{font=small}
\caption{Evaluation of different ABRs on static/dynamic/high SNR channel conditions.}
\resizebox{0.99\textwidth}{!}{
\begin{tabular}{|c|c|c|c|c|c|c|c|c|c|c|c|c|c|c|c|c|}
\hline
\multicolumn{2}{|c}{\textbf{Metrics}} & \multicolumn{5}{|c}{\textbf{Static}} &  \multicolumn{5}{|c}{\textbf{Dynamic}} &  \multicolumn{5}{|c|}{\textbf{High SNR}}\\
\hline
\multicolumn{2}{|c|}{ABR} & FESTIVE & BBA & BOLA & HYB & \textcolor{red}{MPC-S}  & FESTIVE & BBA & BOLA & HYB & \textcolor{red}{MPC-S} & FESTIVE & BBA & BOLA & HYB & \textcolor{red}{MPC-S}  \\
\hline
\multirow{2}{*}{\textbf{PHY}} & BLER & 0.30 & 0.30 & 0.32 & 0.28 & \textbf{0.15}  & 0.397 & 0.412 & 0.415 & 0.32 & \textbf{0.19} & 0.214 & 0.273 & 0.198 & 0.191 &  \textbf{0.067} \\
\cline{2-17} 
  & BER & 0.008 & 0.004 & 0.006 & 0.0006 & \textbf{0.0002}  & 0.014 & 0.008 & 0.009 & 0.012 & \textbf{0.003} & 0.0013 & 0.0027 & 0.0017 & 0.0006 & \textbf{0.000} \\
\hline
\multirow{2}{*}{\textbf{VQA}} & PSNR &  30.50 & 29.52  & 29.37 & 27.52 & \textbf{33.01} & 27.41 & 27.48 & 27.56 & 27.18 & \textbf{32.14} & 34.64 & 33.26 & 33.56 & 36.09 &  \textbf{40.02} \\
\cline{2-17} 
& SSIM &  0.780 & 0.801 &  0.736 &  0.721  &  \textbf{0.86} & 0.74& 0.75& 0.77 & 0.69 & \textbf{0.843} & 0.88 & 0.83 & 0.84 &  0.91  &  \textbf{0.97} \\
\hline
\multirow{4}{*}{\textbf{APP}} & Average Bitrate &  4.49  &  4.27 & 3.99  & 4.20 & \textbf{4.81} & 3.84 & 3.48 & 3.52 & 3.84 & \textbf{3.85} & 4.47 & 4.14 & 4.25 & 4.48 & \textbf{4.99} \\
\cline{2-17} 
& Average Rebuffering & 0.082 & 0.053 &  \textbf{0.008} & 0.056 & 0.018 & 0.169 & 0.019 & \textbf{0.014} & 0.096 & 0.02 & 0.072 & \textbf{0.0082} & 0.019 & 0.032 & 0.015  \\
\cline{2-17} 
& Average QoE & 2.85 & 3.05 & 2.16 & 2.78 & \textbf{4.11} & 1.47 & 1.49 & 1.66 & 1.74 & \textbf{2.01} & 2.65  & 2.74 & 2.69  & 3.44 &  \textbf{4.63} \\
\hline
\end{tabular}}
\label{results}
\end{table*}

\begin{figure*}
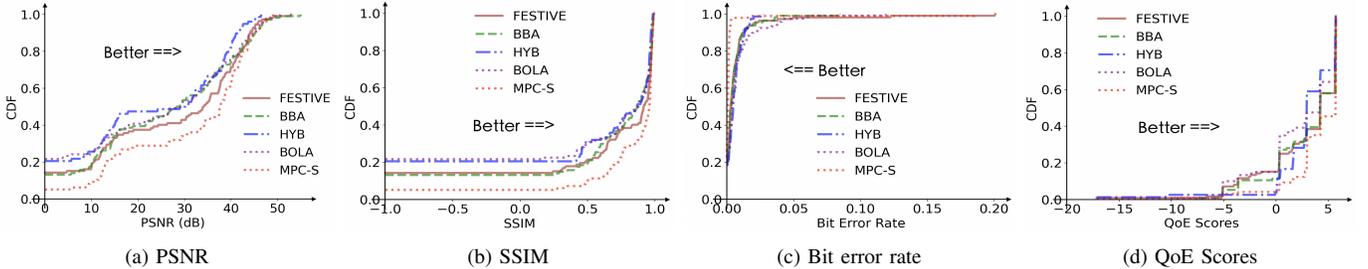

    \begin{subfigure}{0.499\columnwidth}
    \includegraphics[width=\linewidth,height = 0.7\linewidth]{photo/stat_psnr.pdf}
    \caption{PSNR}
    \label{fig:2b}
  \end{subfigure}
  \begin{subfigure}{0.499\columnwidth}
    \includegraphics[width=\linewidth,height = 0.7\linewidth]{photo/stat_ssim.pdf}
    \caption{SSIM}
    \label{fig:2c}
  \end{subfigure}
  \begin{subfigure}{0.499\columnwidth}
    \includegraphics[width=\linewidth,height = 0.7\linewidth]{photo/stat_BIT.pdf}
    \caption{Bit error rate}
    \label{fig:2d}
  \end{subfigure}
    \begin{subfigure}{0.499\columnwidth}
    \includegraphics[width=\linewidth,height = 0.7\linewidth]{photo/stat_QoE.pdf}
    \caption{QoE Scores}
    \label{fig:2e}
    \end{subfigure}
  \vspace{-0.2cm}
    \caption{Evaluation on static link.}
    \vspace{-0.5cm}
    \end{figure*}

\subsection{Simulation Setting}
A 5G PDSCH is configured in \textit{Matlab R2022b} to simulate the multipath propagation and Doppler shift of wireless link. Channel type is set as tapped delay line C and LDPC is deployed for channel coding. The SNR is adjusted between -5 dB and 15 dB and the doppler shift is set from 10 to 50Hz to mimic various channel scenarios, specifically static, dynamic, and high SNR. Specifics of the simulation setup are provided in Table \ref{tab:combined}. It is worth noting that our parameter settings for OLLA offset $\Delta_{\text{up}}, \Delta_{\text{down}}$ and soft-ACK threshold $\Phi$ strictly adhere to our theoretical analysis in Section III.

\textbf{Methods:} For ABRs, we choose four commonly used ABR techniques as baselines. FESTIVE \cite{festive}, a rate-based approach, predicts trhoughput based on historical measurements and selects the closest corresponding bitrate. BBA (buffer-based strategy) makes decisions according to the buffer level. HYB (Hybrid) takes both buffer level and throughput prediction into consideration for better trade-off. BOLA employs Lyapunov optimization to derive an optimal buffer-based strategy. MPC-S is our proposed scheme, which is short for MPC on StreamOptix. As current \textit{MATLAB simulink} takes too long for link simulation, our comparison does not include learning-based ABRs, which trains ABR agent offline based on simulation interactions. 
\begin{table}
\captionsetup{font=small}
\caption{Methods evaluation details, VRA stands for video-aware resource allocation and 3GPP stands for default 3GPP MCS lookup.}
\centering
\resizebox{0.499\textwidth}{!}{
\begin{tabular}{|c|c|c|c|c|c|c|c|} 
\hline
\textbf{Methods} & \textbf{MPC-S} & \textbf{MPC-a} & \textbf{MPC-b} & \textbf{FESTIVE} & \textbf{BOLA} & \textbf{BBA} & \textbf{HYB} \\
\hline
\textbf{APP} & MPC & MPC & MPC & FESTIVE & BOLA & BBA & HYB\\ 
\hline
\textbf{MAC} & VRA & Full grid & Full grid & Full grid & Full grid & Full grid & Full grid \\ 
\hline
\textbf{PHY} & \makecell{Soft link \\ adaptation} & \makecell{Soft link \\ adaptation} & 3GPP & 3GPP & 3GPP & 3GPP & 3GPP \\ 
\hline
\end{tabular}}
\vspace{-0.5cm}
\label{tab:met}
\end{table}

To highlight the gain achieved by our cross-layer design, we configure all ABR baseline's link adaptation scheme to the 3GPP default link adaptation scheme and implement resource allocation with a full grid (RB count stays at $e_{max}$ during throughout transmission). MPC-a and MPC-b are ablation versions for MPC-S that sequentially remove soft link adaptation and video-aware resource allocation strategies. Details about method configuration is illustarted in Table \ref{tab:met}.

  \begin{figure*}
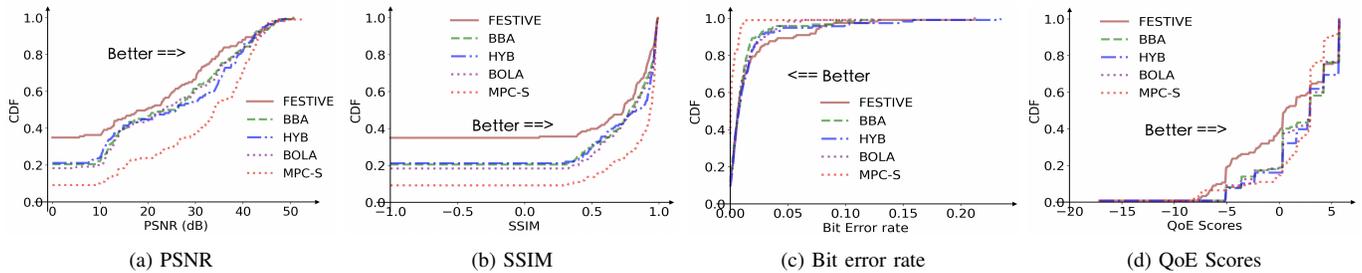

  \begin{subfigure}{0.499\columnwidth}
    \includegraphics[width=\linewidth,height = 0.7\linewidth]{photo/low_psnr.pdf}
    \caption{PSNR}
    \label{fig:2b}
  \end{subfigure}
  \begin{subfigure}{0.499\columnwidth}
    \includegraphics[width=\linewidth,height = 0.7\linewidth]{photo/low_SSIM.pdf}
    \caption{SSIM}
    \label{fig:2c}
  \end{subfigure}
  \begin{subfigure}{0.499\columnwidth}
    \includegraphics[width=\linewidth,height = 0.7\linewidth]{photo/low_BIT.pdf}
    \caption{Bit error rate}
    \label{fig:2d}
  \end{subfigure}
    \begin{subfigure}{0.499\columnwidth}
    \includegraphics[width=\linewidth,height = 0.7\linewidth]{photo/low_QoE.pdf}
    \caption{QoE Scores}
    \label{fig:2e}
    \end{subfigure}
    \vspace{-0.2cm}
    \caption{Evaluation on dynamic link.}
       \vspace{-0.5cm}
    \end{figure*}

    \begin{figure*}
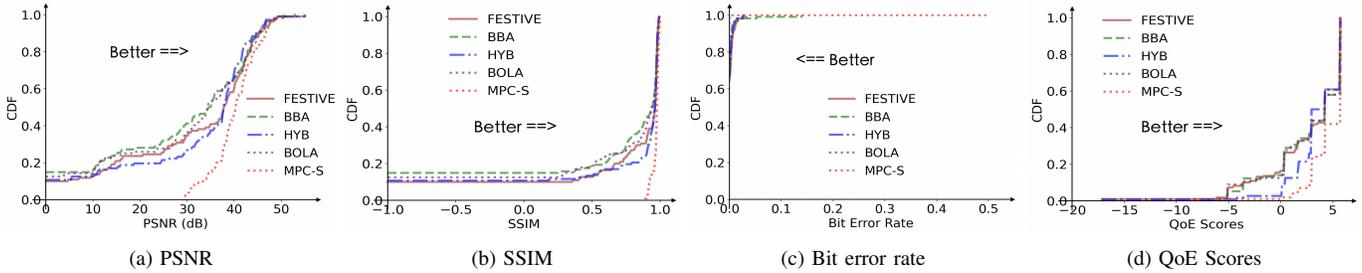

      \begin{subfigure}{0.499\columnwidth}
    \includegraphics[width=\linewidth,height = 0.7\linewidth]{photo/high_psnr.pdf}
    \caption{PSNR}
    \label{fig:2b}
  \end{subfigure}
  \begin{subfigure}{0.499\columnwidth}
    \includegraphics[width=\linewidth,height = 0.7\linewidth]{photo/high_SSIM.pdf}
    \caption{SSIM}
    \label{fig:2c}
  \end{subfigure}
  \begin{subfigure}{0.499\columnwidth}
    \includegraphics[width=\linewidth,height = 0.7\linewidth]{photo/high_BIT.pdf}
    \caption{Bit error rate}
    \label{fig:2d}
  \end{subfigure}
    \begin{subfigure}{0.499\columnwidth}
    \includegraphics[width=\linewidth,height = 0.7\linewidth]{photo/high_QoE.pdf}
    \caption{QoE Scores}
    \label{fig:2e}
    \end{subfigure}
      \vspace{-0.2cm}
    \caption{Evaluation on High SNR link.}
      \vspace{-0.2cm}
  \end{figure*}

\begin{figure*}
    \centering
    \begin{minipage}[t]{0.625\linewidth}
        \centering
        \includegraphics[width=\linewidth, height = 0.29\linewidth]{photo/rate.pdf}
         \caption{Chunk bitrate v.s. throughput on Static/Dynamic/High SNR.}
         \label{video trace}
    \end{minipage}
    \begin{minipage}[t]{0.365\linewidth}
        \centering
        \includegraphics[width=\linewidth, height = 0.47\linewidth]{photo/utilization.pdf}
        \caption{Evaluation on rate utilization.}
        \label{rate uti}
    \end{minipage}
    \vspace{-0.6cm}
\end{figure*}

\textbf{Dataset:} We employ the Waterloo Dataset \cite{waterloo} as the streaming source, which comprises twenty videos (including documentaries, sports, and gaming). These videos, each 10 seconds long, are offered in 11 distinct quality levels ranging from 240p to 1440p, encoded in H.264 format. A bitrate ladder is constructed from these 11 quality levels, set at [750, 1750, 2350, 3000, 4300, 7000] kbps. Each video is segmented into 2-second chunks using \textit{FFmpeg} with a Constant Rate Factor of 23. Every simulation run consists of one hundred video segments, culminating in a total viewing time of 200 seconds.

\textbf{Evaluation Metrics:} To evaluate ABR performance at APP, we analyze the video bitrate, rebuffering, and QoE scores. We compute the PSNR (Peak Signal-to-Noise Ratio) and SSIM (Structural Similarity Index) to assessing the quality of delivered videos. Additionally, we meaure the BLER, bit error rate, and rate utilization during video transmission to assess the performance of wireless link. Rate utilization for chunk $i$ is defined as $\frac{t_{i+1} - t_{i}}{T_{\rm chunk}}$, $t_{i}$ represents the time to start downloading chunk $i$. Rate utilization of 1 indicates the perfect match of video bitrate and transmission throughput. A value exceeding 1 suggests underutilization of the available throughput, whereas a value less than 1 indicates overutilization.

In the following sections, we first provide an overall comparison between MPC-S and other baseline methods across three different link conditions. We then analyze the effects of PHY-aware throughput prediction, video-aware resource allocation (VRA) and soft link adaptation separately to illustrate the gain of our cross-layer design. Finally, we verify that our cross-layer design could serve as a plug-in solution which could enhance the performance of any ABR. 

% In the following, we evaluate MPC-S with all baselines across three different link environments: static, dynamic, and high SNR (set according to Table \ref{tab:combined}).

% The SNR values for static and dynamic links are set to vary between -5 to 5 dB, while for the high SNR link, the SNR are set from 5 to 15 dB. The rate of change in the link is determined by the Doppler shift, with a higher Doppler shift leading to more rapid changes in channel conditions. The Doppler shift for both the static and high SNR are set at 10 Hz, while for the dynamic link, it is set at 50 Hz.

\subsection{Results on StreamOptix Streaming}
In the following, we evaluate MPC-S with other baseline methods across three different link environments: static, dynamic, and high SNR (configured according to Table \ref{tab:combined}). The cumulative distribution function (CDF) of SSIM, PSNR, bit error rate and QoE are illustrated in Fig. 5, 6 and 7. Extensive meaurements of all metrics are listed in Table \ref{results}.

\textbf{Static Scenario:} As shown in Fig. 5, MPC-S achieves a significant improvement in objective quality assessments (PSNR, SSIM) compared to other baselines. The improvement is attributed to our soft link adaptation, as confirmed by the bit rate error curve. Specifically, the intersections at y-axis for the SSIM and PSNR curves suggest the ratio of video decoding failures. This occurs when the essential metadata needed for decoding is corrupted during the transmission of the link. Remarkably, the MPC-S approach shows the lowest occurrence of decoding errors, around 0.05, subtly emphasizing the decrease in errors due to soft link adaptation. In contrast, HYB and BOLA show error rates as high as 20\%. The elimination of error rate can be attributed to the incorporation of our soft link adaptation, which enhances the adaptability of link transmission to prevailing channel conditions. 

% QoE scores is observed with MPC-S, showcasing a 34\%
% enhancement compared to other ABR methods. 

\textbf{Dynamic Scenario:} In dynamic link, the frequent link variations lead to greater variations in throughput. According to Fig. 6, the SSIM and PSNR curves show nearly a twofold rise in the bit error rate across all methods. Results in Table \ref{results} demonstrate that the MPC-S improves QoE by 14\% relative to HYB, which is the most effective among all baselines. SSIM and PSNR of MPC-S reaches 0.843 and 32.14 dB, respectively.

\textbf{High SNR:} The high SNR link simulates the good channel condition with large throughput and small variations. In this case, the increase of SNR leads to a greater link capacity, resulting in improved performance across all methods. 
% While the bit error rate may not directly indicate the effect on video quality, since video degradation is influenced by both the proportion of errors and their relative locations, the analysis of three network scenarios clearly demonstrates a substantial enhancement in video quality when the bit error rate decreases by an order of magnitude. Specifically, the intersects of the SSIM and PSNR curves for the MPC-S algorithm are zero, indicating that with link control in place, all transmitted videos are accurately decoded. 
As observed in Fig. 7, MPC-S has completely eradicated video decoding errors, significantly enhancing the objective video quality, with PSNR achieving 40.02 dB and SSIM attaining 0.97. Despite the improvement in performance across all methods, MPC-S continues to show superior performance in nearly all metrics. According to Table \ref{results}, MPC-S outperforms the runner-up method with a 34\% enhancement in QoE, an 11\% improvement in PSNR, and an 8\% rise in SSIM.

%  \begin{figure}
%       \begin{subfigure}{0.325\columnwidth}
%       \begin{center}
%     \includegraphics[width=0.95\linewidth,height = 0.7\linewidth]{photo/trace_st.pdf}
%     \label{fig:2c}
%     \end{center}
%   \end{subfigure}
%       \begin{subfigure}{0.325\columnwidth}
%       \begin{center}
%     \includegraphics[width=0.95\linewidth,height = 0.7\linewidth]{photo/trace_dy.pdf}
%     \label{fig:2c}
%     \end{center}
%   \end{subfigure}
%     \begin{subfigure}{0.325\columnwidth}
%       \begin{center}
%     \includegraphics[width=0.95\linewidth,height = 0.7\linewidth]{photo/trace_hi.pdf}
%     \label{fig:2c}
%     \end{center}
%   \end{subfigure}
%   \caption{Chunk bitrate v.s. bandwidth on Static/Dynamic/High SNR channels}
%   \label{video trace}
%   \end{figure}

% \begin{figure}[!h]
%     \centering
%     \includegraphics[width = \linewidth,height = 0.37 \linewidth]{photo/utilization.pdf}
%     \caption{Rate utilization on Static/Dynamic/High SNR channels}
%         \label{rate uti} 
% \end{figure}

\textbf{Rate Utilization:} Fig. \ref{rate uti} exhibits rate utilization of all methods. The results reveal that MPC-S has effectively manage the resource to maximize rate utilization, with recorded values of 0.95, 0.92, and 0.99. In contrast, other methods exhibit a 15\%-20\% shortfall due to the lack of video bitrate allocation. MPC-S also shows the smallest variance when compared to other methods. To visualize the effects of video-aware resource allocation (VRA), we construct a video bitrate trace with five chunks at rates of 3000kbps, 4300kbps, and 5800kbps. We adjust the SNR of wireless link to ensure that the link capacity marginally exceeds the bitrate in the constructed trace. Fig. \ref{video trace} illustrates the throughput measurements for each chunk, both with and without VRA implementation. The results indicate improved rate utilization when VRA is applied. Channel throughput demonstrates enhanced stability in both static and elevated SNR scenarios. Given that VRA is influenced by BLER prediction (Eq. (24)), the variance in dynamic link conditions increases as BLER becomes more unpredictable.

\section{Analysis and ablation study}
\subsection{PHY-aware Throughput Prediction}
% In this section, we explore the effects of PHY-aware link capacity prediction on ABR (MPC in our design) performance. 
As previously mentioned, our cross-layer scheme provides regularly sampled link capacities at PHY, which allows for more accurate predictions. Fig. \ref{MPC-time1} shows the impact of regularly sampled throughput prediction at different granularities. It is evident from the illustration that as the granularity decreases, the predicted throughput becomes increasingly conservative. Conversely, finer granularities may lead to higher overhead and more aggressive predictions. To determine the optimal prediction granularity that trade off these effects, we assess MPC's performance using granularities ranging from 150ms to 900ms. Throughput is predicted using harmonic mean of past 8 samples \cite{mpc}. The results presented in Table \ref{tab:mpc} and Fig. \ref{MPC-time2} indicate that with increasing granularity, the MPC's average bitrate consistently drops, likely due to more conservative predictions. At finer granularities, the bitrates are notably high, leading to significant rebuffering to deteriorate QoE. With coarser granularities, there is a bitrate reduction accompanied by less rebuffering, which helps mitigate the impact of lower bitrates. An optimal QoE is achieved at a granularity of approximately 600ms. After this threshold, additional reductions in rebuffering do not offset the decline in bitrate. Therefore, we set the granularity of throughput prediction as 600ms throughout our study. In Fig. \ref{MPC-time2}, the horizontal line depicts the performance of naive MPC \cite{mpc}, which predicts future throughput at chunk level. The results indicate naive MPC's QoE is approximately 30\% worse compared to that of MPC with 600ms prediction granularity.

\begin{figure}
\begin{subfigure}{0.49\columnwidth}
\includegraphics[width=\linewidth,height = 0.68\linewidth]{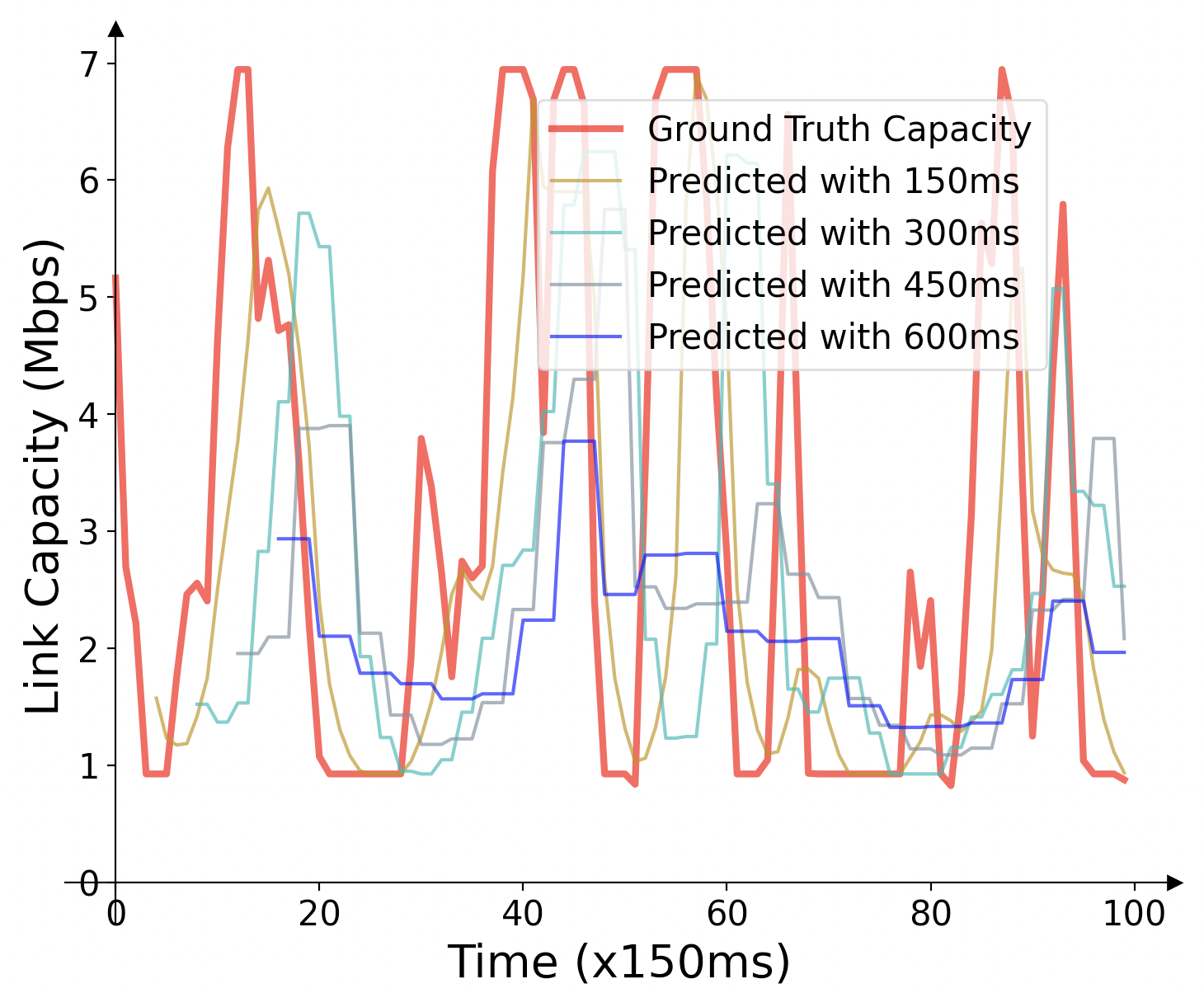}
  \caption{Predict link capacities with different prediction intervals}
 \label{MPC-time1}
\end{subfigure}
 \begin{subfigure}{0.5\columnwidth}
\includegraphics[width=\linewidth,height = 0.67\linewidth]{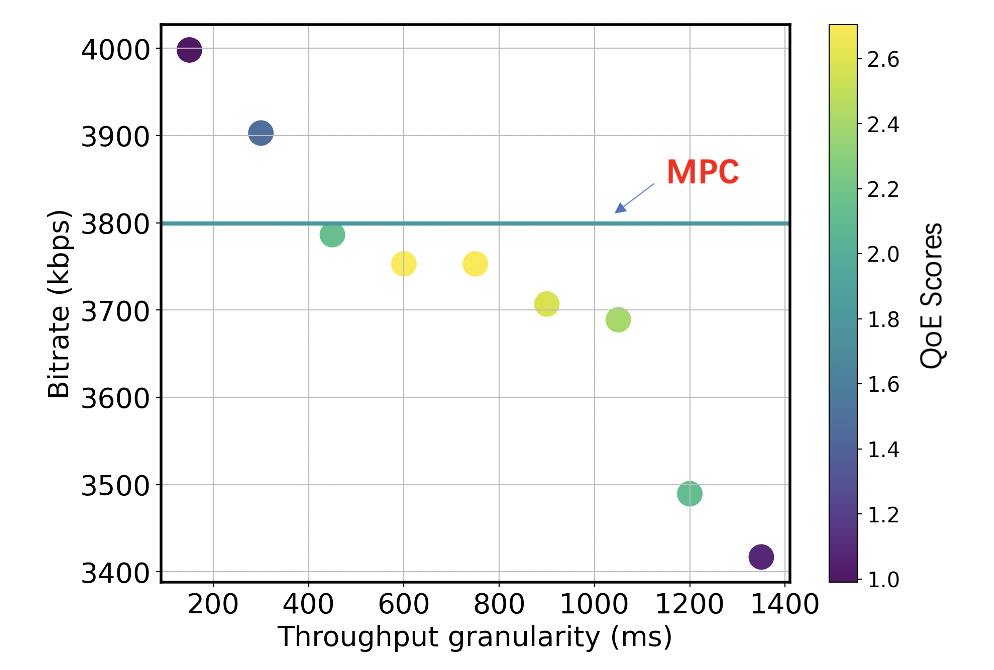}
  \caption{MPCs with different prediction intervals}
   \label{MPC-time2}
\end{subfigure}
\caption{Evaluation on PHY-aware throughput prediction.}
\end{figure}

\begin{table}[t]
\caption{Performance of MPCs on different prediction granularities.}
\centering
\resizebox{0.495\textwidth}{!}{
\begin{tabular}{|c|c|c|c|c|c|c|} 
\hline
\textbf{Granularity (ms)} & \textbf{150} & \textbf{300} & \textbf{450} & \textbf{600} & \textbf{750} & \textbf{900} \\
\hline
MSE Loss on prediction & \textbf{3.17} & 6.94 & 5.98 & 3.61 & 5.20 & 4.86\\ 
\hline
QoE Scores & 0.99 & 1.47 & 2.13 & \textbf{2.70} & 2.58 & 2.44\\ 
\hline
Bitrate (Mbps) & 3.998 & 3.903 & 3.789 & \textbf{3.753} & 3.707 & 3.689 \\ 
\hline
Rebuffer (s) & 0.205 & 0.122 & 0.058 & \textbf{0.017} & 0.04 & 0.059 \\ 
\hline
 \end{tabular}}
 \label{tab:mpc}
 \end{table}

\begin{figure}[t]
\begin{subfigure}{0.49\columnwidth}
\includegraphics[width=0.99\linewidth,height = 0.69\linewidth]{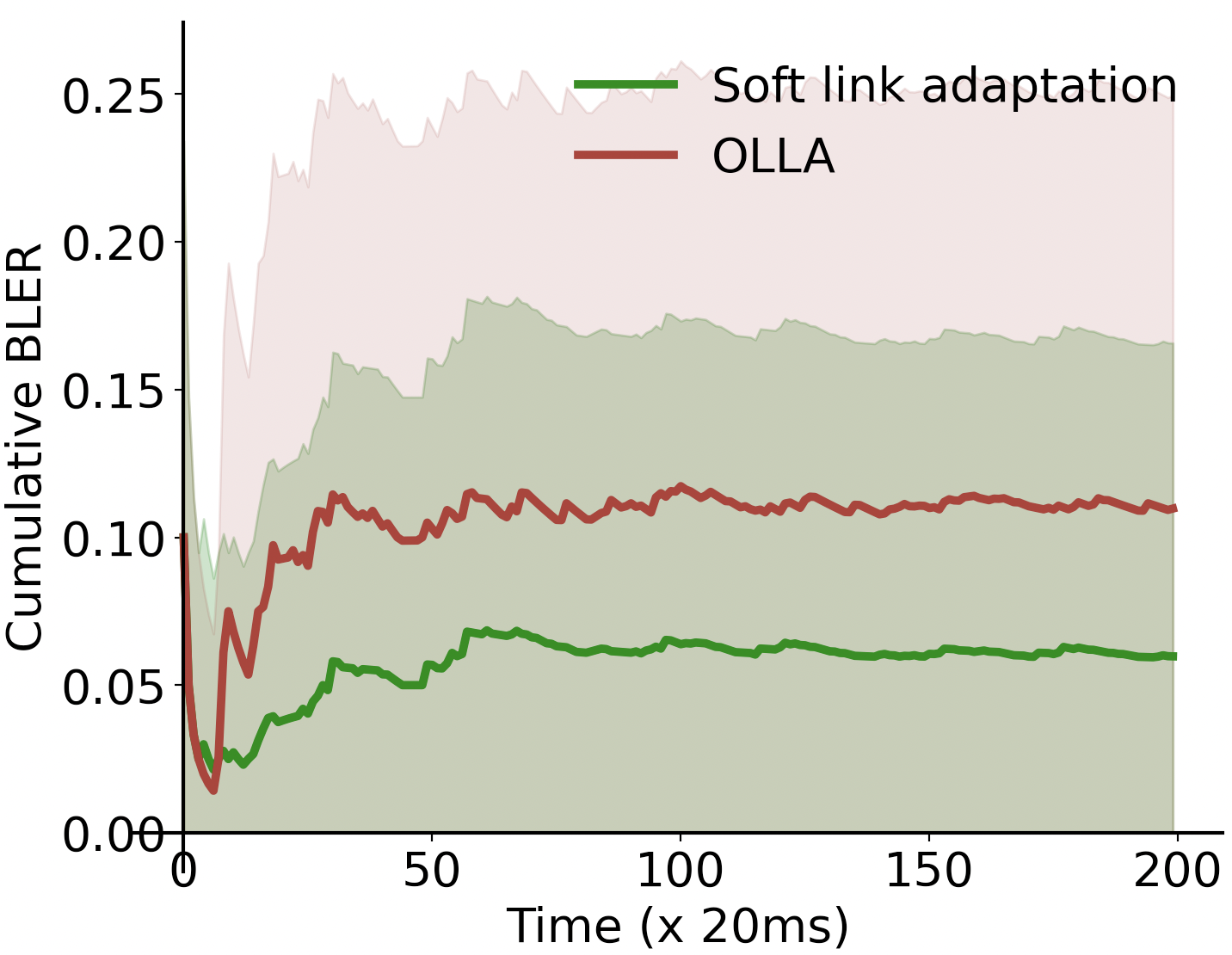}
  \caption{Visualization on soft link adaptation convergence}
 \label{figconvergence}
\end{subfigure}
 \begin{subfigure}{0.5\columnwidth}
\includegraphics[width=0.9\linewidth,height = 0.67\linewidth]{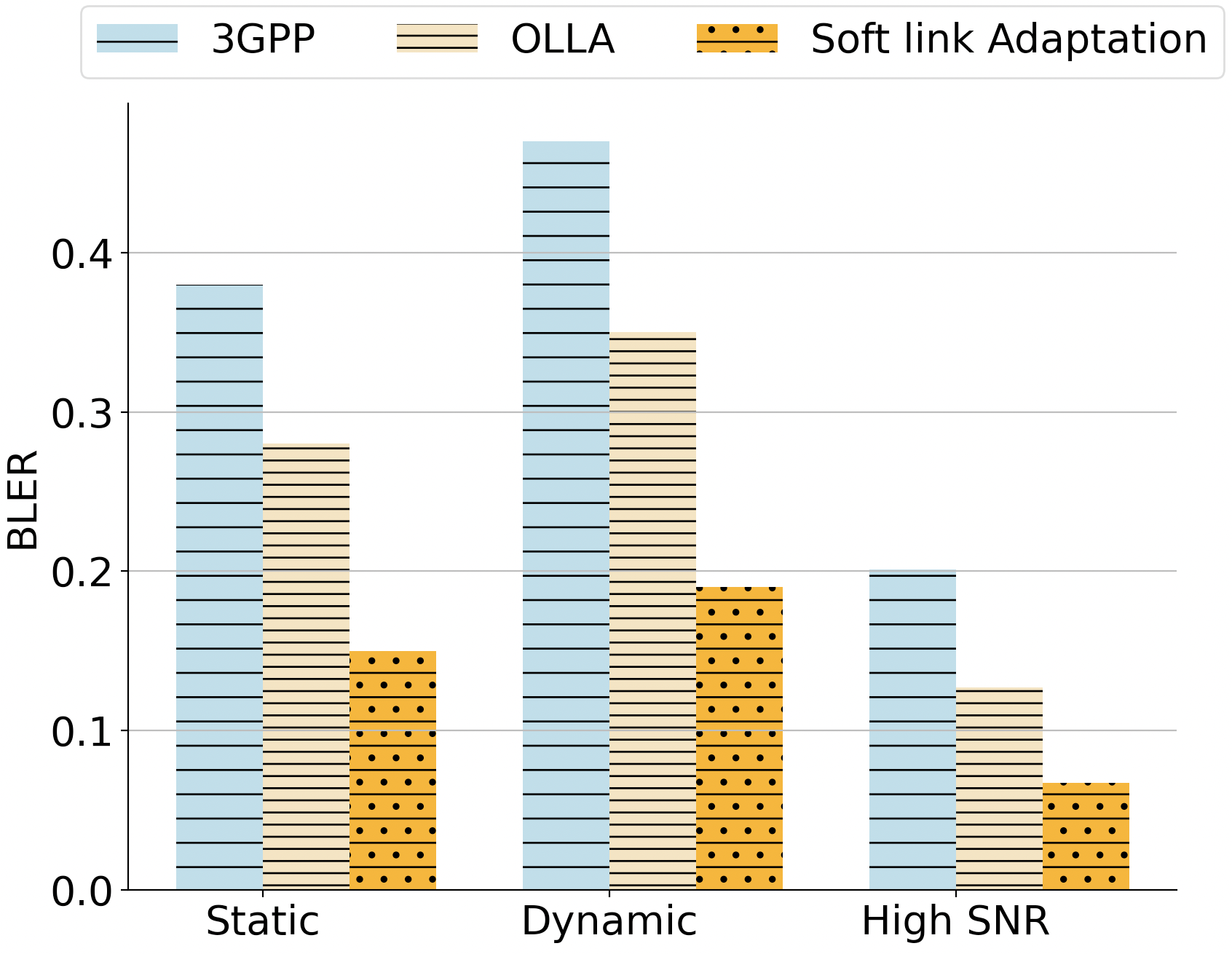}
  \caption{BLERs assessed at different link conditions}
   \label{bler}
\end{subfigure}
\caption{Evaluation on soft link adaptation.}
\end{figure}

\begin{table}[t]
\caption{Comparison of BLER performance at different SNR levels.}
\centering
\resizebox{0.499\textwidth}{!}{ % 调整表格宽度
\begin{tabular}{|c|c|c|c|c|c|c|c|} 
\hline
\textbf{BLER target} & \diagbox{\textbf{Schemes}}{\textbf{SNR}} & \textbf{-5 dB} & \textbf{0 dB} & \textbf{5 dB} & \textbf{10 dB} & \textbf{15 dB} & \textbf{20 dB} \\
\hline
\multirow{3}{*}{\textbf{0.1}} & 3GPP & 0.7 & 0.57 & 0.23 & 0.05 & 0.02 & 0.007\\ 
\cline{2-8}
& OLLA & \textbf{0.54} & \textbf{0.18} & \textbf{0.11} & \textbf{0.03} & \textbf{0.009} & \textbf{0.006}\\ 
\cline{2-8}
& \makecell{Soft Link \\Adaptation} & 0.55 & 0.17 & 0.07 & 0.02 & 0.005 & 0.004 \\ 
\hline
\multirow{3}{*}{\textbf{0.05}} & 3GPP & 0.7 & 0.57 & 0.23 & 0.05 & 0.02 & 0.007\\ 
\cline{2-8}
& OLLA & \textbf{0.54} & 0.18 & 0.07 & 0.02 & 0.007 & 0.005\\ 
\cline{2-8}
& \makecell{Soft Link \\Adaptation} & 0.55 & \textbf{0.17} & \textbf{0.052} & \textbf{0.015} & \textbf{0.003} & \textbf{0.0014} \\ 
\hline
\end{tabular}}
\label{tab:3}
\end{table}

\subsection{Soft Link Adaptation}
As shown in Fig. \ref{figconvergence}, soft link adaptation has effectively achieved BLER convergence at chunk level, which reaches stability at around 1 second. This confirms the chunk-level convergence as discussed in Section III B of our theoretical analysis. To demonstrate the average BLER level achieved by soft link adaptation, we assess the average BLER of soft link adaptation, OLLA, and the default 3GPP MCS lookup (3GPP) under different SNR conditions. As illustrated in Fig. \ref{bler}, soft link adaptation surpasses the other two schemes, attaining an almost 50\% reduction in BLER level across all link conditions. Table \ref{tab:3} presents a comprehensive analysis of BLER at two different BLER targets. The results indicate that the BLER convergence is reached at around 5dB SNR under soft link adaptation, significantly lower than the other two. Extensive assessments indicate that soft link adaptation outperforms the other two methods, leading to a more uniform BLER level. This consistency ensures more stable throughput and reduce the bit error rate, ultimately enhancing the quality of the delivered videos.
\begin{figure*}
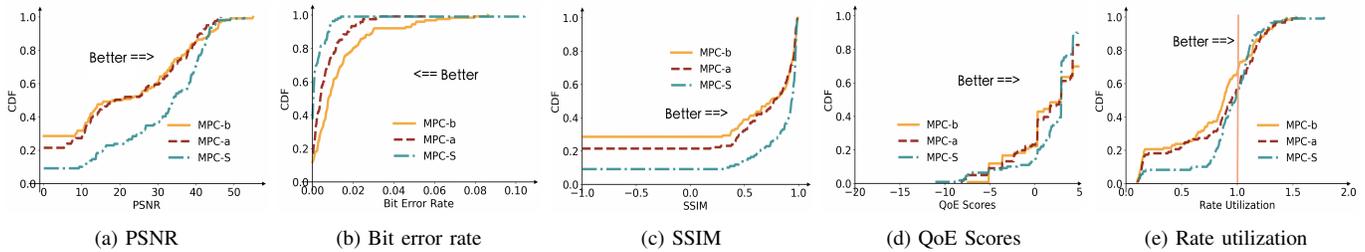

  \begin{subfigure}{0.395\columnwidth}
    \includegraphics[width=\linewidth,height = 0.8\linewidth]{photo/abla1.pdf}
    \caption{PSNR}
    \label{fig:2b}
  \end{subfigure}
  % \begin{subfigure}{0.49\columnwidth}
  %   \includegraphics[width=\linewidth,height = 0.84\linewidth]{photo/abla2.pdf}
  %   \caption{Color Artfs.}
  %   \label{fig:2c}
  % \end{subfigure}
  \begin{subfigure}{0.395\columnwidth}
    \includegraphics[width=\linewidth,height = 0.8\linewidth]{photo/abla4.pdf}
    \caption{Bit error rate}
    \label{fig:2d}
  \end{subfigure}
    \begin{subfigure}{0.395\columnwidth}
    \includegraphics[width=\linewidth,height = 0.8\linewidth]{photo/abla3.pdf}
    \caption{SSIM}
    \label{fig:2e}
    \end{subfigure}
      \begin{subfigure}{0.395\columnwidth}
      \begin{center}
    \includegraphics[width=\linewidth,height = 0.8\linewidth]{photo/abla_QoE.pdf}
    \caption{QoE Scores}
    \label{fig:2c}
    \end{center}
  \end{subfigure}
       \begin{subfigure}{0.395\columnwidth}
      \begin{center}
    \includegraphics[width=\linewidth,height = 0.8\linewidth]{photo/abla_uti.pdf}
    \caption{Rate utilization}
    \label{fig:2c}
    \end{center}
  \end{subfigure}
  % \end{figure*}
  % \begin{figure*}[!h]
  %   \centering
  %       \begin{subfigure}{0.49\columnwidth}
  %     \begin{center}
  %   \includegraphics[width=0.99\linewidth]{photo/rebuf_all.pdf}
  %   \caption{Rebuffer}
  %   \label{fig:2c}
  %   \end{center}
  % \end{subfigure}
  % \begin{subfigure}{0.49\columnwidth}
  %     \begin{center}
  %   \includegraphics[width=0.93\linewidth]{photo/qoeyou.pdf}
  %   \caption{QoE Scores}
  %   \label{fig:2c}
  %   \end{center}
  % \end{subfigure}
  % \begin{subfigure}{0.49\columnwidth}
  %     \begin{center}
  %   \includegraphics[width=0.93\linewidth]{photo/PSNRyou.pdf}
  %   \caption{PSNR}
  %   \label{fig:2c}
  %   \end{center}
  % \end{subfigure}
  % \begin{subfigure}{0.49\columnwidth}
  %     \begin{center}
  %   \includegraphics[width=0.936\linewidth]{photo/ssimyou.pdf}
  %   \caption{SSIM}
  %   \label{fig:2c}
  %   \end{center}
  % \end{subfigure}
  \caption{Ablation study of cross-layer design.}
  \label{ablation}
  \vspace{-0.2cm}
\end{figure*}
\subsection{Incremental Ablation Study}
To delve deeper into the benefits of soft link adaptation and video aware resource allocation, we conduct an incremental ablation analysis on MPC\footnote{The experimental setup for ablation study is based on the dynamic configuration.}. In this analysis, MPC-b denotes the variant lacking both soft link adaptation and video aware resource allocation, whereas MPC-a includes only soft link adaptation. As shown in Fig. \ref{ablation}, the average bit error rate significantly drops from 0.4 to below 0.2 with the implementation of soft link adaptation. Moreover, the addition of VRA to the soft link adaptation leads to a further reduction in the bit error rate, which falls from an average of 0.05 to nearly zero. The rate utilization of MPC is significantly improved by VRA, with MPC-S reaching a mean value of 0.95, compared to 0.8 for MPC-a and 0.75 for MPC-b (Fig. 12e). Additionally, MPC-S exhibits significant improvements in terms of PSNR, SSIM, and QoE.
\begin{figure*}[!h]
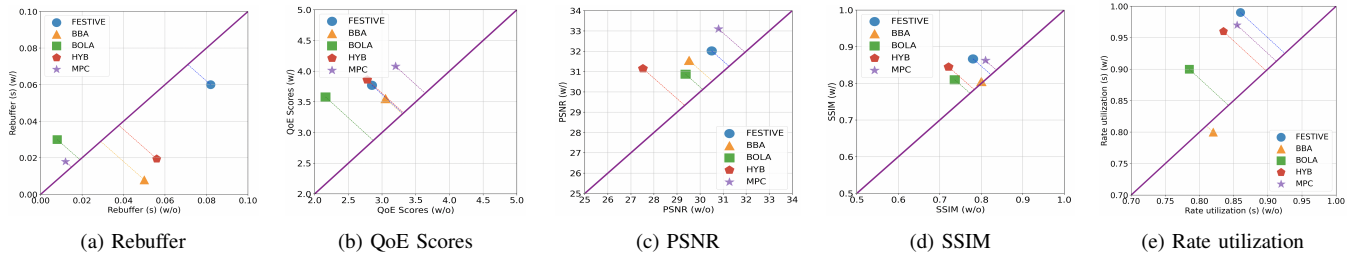

   \centering
        \begin{subfigure}{0.395\columnwidth}
      \begin{center}
    \includegraphics[width=0.95\linewidth, height = 0.8\linewidth]{photo/rebuf_all.pdf}
    \caption{Rebuffer}
    \label{fig:2c}
    \end{center}
  \end{subfigure}
  \begin{subfigure}{0.395\columnwidth}
      \begin{center}
    \includegraphics[width=0.91\linewidth, height = 0.8\linewidth]{photo/qoeyou.pdf}
    \caption{QoE Scores}
    \label{fig:2c}
    \end{center}
  \end{subfigure}
  \begin{subfigure}{0.395\columnwidth}
      \begin{center}
    \includegraphics[width=0.93\linewidth,height = 0.8\linewidth]{photo/PSNRyou.pdf}
    \caption{PSNR}
    \label{fig:2c}
    \end{center}
  \end{subfigure}
  \begin{subfigure}{0.395\columnwidth}
      \begin{center}
    \includegraphics[width=0.94\linewidth,height = 0.8\linewidth]{photo/ssimyou.pdf}
    \caption{SSIM}
    \label{fig:2c}
    \end{center}
  \end{subfigure}
    \begin{subfigure}{0.395\columnwidth}
      \begin{center}
    \includegraphics[width=0.94\linewidth,height = 0.82\linewidth]{photo/utiyou.pdf}
    \caption{Rate utilization}
    \label{fig:2c}
    \end{center}
  \end{subfigure}
\caption{Benefits of our cross-layer design on different ABRs.}
\vspace{-0.5cm}
\label{any}
  \end{figure*}

\section{Discussion and Challenge}
In this study, we present the closed-loop cross-layer optimization platform that integrates interactions across the PHY, MAC, and APP layers. Utilizing our cross-layer transmission framework, any ABR can adopt soft link adaptation and VRA.\par
\textbf{Benefits of StreamOptix:} To demonstrate the efficacy of our cross-layer design for all ABRs, we have included each ABR with soft link adaptation and VRA in our study for comprehensive assessment. The findings are depicted in Fig. \ref{any}a-e. The horizontal axis depicts the metrics evaluated using ABRs that include both soft link adaptation and VRA, whereas the vertical axis depicts the metrics assessed using ABRs lacking both cross-layer designs. Thus, the divergence from the diagonal line visually indicates the advantages gained from the cross-layer design. It is evident that, with the exception of BBA, which exhibits no notable variation in SSIM, all metrics show improvements across different ABRs, especially HYB. MPC outperforms all other ABRs in each scenario, validating our selection of MPC as the ABR. StreamOptix is also capable of simulating a range of distortions that replicate real-world transmission situations. Visualizations of these distortions can be found in Appendix C. StreamOptix offers access to both the distorted video bitstreams and the decoded videos. This assists future researchers in collecting additional distorted transmission streams to further explore the impact of various transmission configurations on video quality.\par

\textbf{System Complexity:} A major hurdle in evaluating StreamOptix is its computational complexity. The transmission of each TB involves channel convolution to account for multipath effects and doppler shift variations, causing a 2-second video chunk simulation to take several minutes. The entire assessment is performed on an Intel i9 processor, requiring approximately one and a half months to complete, thereby constraining the extension of learning-driven ABRs on this system. Future studies will explore the distribution of bit error patterns in simulated video streams across our wireless links. This exploration will help pinpoint distortion patterns within this distribution, enhancing our ability to simulate impairments more accurately in video stream transmission and possibly increasing the efficiency of link simulations. Although our design focuses on a single-user scenario, future research could examine multi-user case by integrating video-aware resource allocation with the existing strategies for multi-user resource allocation as noted in Section III E.

\section{Conclusion}
In this paper, we proposed a cross-layer video delivery scheme, StreamOptix, to effectively optimize video delivery. Our proposed three-level closed-loop framework optimizes video delivery across PHY, MAC and APP. In PHY, we propose soft link adaptation that implements soft-ACK feedback to improve the robustness of link adaptation. In MAC, resource allocation is further improved by assigning resource blocks under bitrate constraint to ensure maximum rate utilization. Then, in APP, we redesigned the MPC algorithm for ABR by utilizing link capacities measured from PHY to improve bitrate decision accuracy. Extensive assessments across diverse network conditions reveal substantial improvements of our design in both subjective and objective quality metrics. Furthermore, StreamOptix provides access to distorted transmission streams, enabling a more thorough quality evaluation and supplying multi-modal data for enhanced optimization.

% === IV. Transistor Class-F inv Rectifier ========================================
% =================================================================================

% or


\appendix
\section{Soft link adaptation algorithm}
\begin{algorithm}
  \caption{Soft Link Adaptation Algorithm}
  \label{alg1}
  \begin{algorithmic}
  \STATE \textbf{Initialize:} $\rm MCS-Table: \gamma \rightarrow \tau_{1,2,...15}$; Soft-ACK threshold $\Phi$
  \STATE \textbf{Initialize:} $\delta_t = \delta_0$, $\Delta_\text{down}$, $\Delta_\text{up}$ (Specify initial values)
  \STATE \textbf{Input:} Estimated SNR $\gamma_t$ from CSI; HARQ feedback flag from UE
  \FOR{each transmission $t$}
    \STATE $\gamma_t = \gamma_t - \delta_t$
    \STATE select $\tau_t$ from  $\gamma \rightarrow \tau_{1,2,...15}$ based on $\gamma_t$
    \IF{flag = ACK}
      \STATE estimate the block error probability $\hat{\eta_t}$
      \IF{$\hat{\eta_t} \leq \Phi$}
        \STATE flag = HIGHACK
      \ELSE
        \STATE flag = LOWACK
      \ENDIF
    \ENDIF
    \IF{flag = HIGHACK}
      \STATE $\delta_t \leftarrow \delta_t - \Delta_\text{down}$
    \ELSIF{flag = LOWACK \textbf{or} flag = NACK}
      \STATE $\delta_t \leftarrow \delta_t + \Delta_\text{up}$
    \ENDIF
  \ENDFOR
  \end{algorithmic}
\end{algorithm}

\section{Proof of the second Banach fixed point condition}
\subsection{OLLA}
% The second Banach fixed point condition states that
% \begin{equation}
%     | T^{\prime}(\delta)|\leq 1, \forall \delta \in [\delta_{\text{min}}, \delta_{\text{max}}].
% \label{sbf}
% \end{equation}
According to \cite{banach}, the second Banach fixed point theorem asserts that the magnitude of the derivative of the recurrence function must be less than 1, that is:
\begin{equation}
 \left|T^{\prime}\left(\delta \right)\right|<1, \forall \delta \in\left[\delta_{\min }, \delta_{\max }\right].
 \label{sbf}
\end{equation}
Recall that the recurrence function for OLLA offset $\delta_t$ is
\begin{equation}
   T(\delta) = \delta + \eta_t \Delta_{\rm up} - (1- \eta_t) \Delta_{\rm down}.
\end{equation}
If we treat the average BLER as a function of $\delta$, then we have  
\begin{equation}
T^{\prime}\left(\delta \right)=1+\left(\Delta_{\rm up}+\Delta_{\rm down} \right) \cdot \eta_t^{\prime}\left(\delta \right).
\end{equation}
Note that the average BLER $\eta_t$ shows a negative correlation with $\delta$. This is because a higher $\delta$ can lead to an overestimated SNR, thereby increasing the chance of block decoding errors. According to \cite{eOLLA}, $\eta_t$ can be fitted via offline link simulations as $\eta_t(\delta) = 1/(1+e^{-\alpha_0\delta-\alpha_1})^s$, where $\alpha_0$, $\alpha_1$ and $s$ are fitting parameters and $s$ controls the slope. 
In order to satisfy the second Banach fixed point condition in Eq. (\ref{sbf}), since the derivative of $\eta_t \left(\delta \right)$ is negative, we have that 
\begin{equation}
\begin{aligned}
\left(\Delta_{\rm up}+\Delta_{\text {down}}\right) \cdot\left|\eta_t^{\prime}\left(\delta \right)\right|<2 \\ \Rightarrow\left|\eta_t^{\prime}\left(\delta \right)\right|<\frac{2}{\left(\Delta_{\rm up}+\Delta_{\text {down}}\right)}.
\label{2nd}
\end{aligned}
\end{equation}

We are going to find the maximum of $\eta_t^{\prime}\left(\delta\right)$, i.e., the values for which $\eta_t^{\prime \prime}\left(\delta\right)=0$, to ensure that the second Banach fixed point condition in Eq. (\ref{sbf}) is fulfilled. It can be shown that the maximum value of $\eta_t^{\prime}\left(\delta\right)$ can be achieved when the offset $\delta$ takes a value equal to
\begin{equation}
\delta_0 =\frac{1}{\alpha_1} \cdot\left(\ln \left(\frac{1}{s}+\alpha_0\right)\right).
\end{equation}
Therefore, the maximum value of $\left|\eta_t^{\prime}\left(\delta \right)\right|$ is
\begin{equation}
|\eta_t^{\prime}(\delta)|=\left|\alpha_1\right| \cdot \frac{1}{(1+1 / s)^{s+1}}.
\end{equation}
An upper bound on $\eta_t^{\prime}\left(\delta \right)$ can thus be found by setting $s \rightarrow \infty$, which is
\begin{equation}
\left|\eta_t^{\prime}(\delta)\right|=\frac{\left|\alpha_1\right|}{e} \text {. }
\label{last}
\end{equation}
According to \cite{eOLLA}, $\alpha_1$ is set to be -1.11, Combining (\ref{2nd}) and (\ref{last}) yields
\begin{equation}
\Delta_{\text {down}} + \Delta_{\text{up}} \leq \frac{2e}{1.11}.
\end{equation}
Therefore, convergence is guaranteed if $\Delta_{\rm up} + \Delta_{\rm down}$ does not exceed a certain limit $\frac{2e}{1.11}$.
\subsection{Soft link adaptation}
Recall that the recurrence function for soft link adaptation is
\begin{equation}
\begin{aligned}
T(\delta) = \delta &+ \eta_t\Delta_{\rm up}
- (1-\eta_t)\eta(\Phi)\Delta_{\rm down} \\
&+ (1-\eta_t)(1-\eta(\Phi))\Delta_{\rm up}.
\label{soft dynamic}
\end{aligned}
\end{equation}
Similarly, by taking the derivative of both sides, we have
\begin{equation}
    T'(\delta) = 1 +(\eta(\Phi)(\Delta_{\rm up}+\Delta_{\rm down})) \cdot \eta'_t(\delta).
\end{equation}
Since $\eta(\Phi) \leq 1$, in order to satisfy second Banach fixed point condition in Eq. (1), we have
\begin{equation}
    |\eta'_t(\delta)| < \frac{2}{\eta(\Phi)(\Delta_{\rm down} + \Delta_{\rm up})}.
\end{equation}
In the main text, it has been shown that the requirement for the first Banach fixed point condition to be satisfied by the soft link adaptation is
\begin{equation}
\eta(\Phi) \geq \frac{1}{1+\frac{\Delta_\text{\rm down}}{\Delta_\text{\rm up}}}.
\end{equation}
Refering to Eq. (\ref{last}), which establishes the maximum value of $\eta'_t(\delta)$ as $|\alpha_1|/e$, and combining Eq. (11) and (12), the parameters $\Delta_{\rm down}$ and $\Delta_{\rm up}$ must fulfill
\begin{equation}
    \Delta_{\rm up} \leq \frac{2e}{1.11}.
\end{equation}
Therefore, convergence is guaranteed if $\Delta_{\text{up}}$ does
not exceed $\frac{2e}{1.11}$.
\section{Video adaptation flow on MPC-S}
\begin{algorithm}
  \caption{Video adaptation flow on MPC-S}
  \label{alg1}
  \begin{algorithmic}
  \STATE \textbf{Initialize:} Prediction window length $M$; Throughput prediction interval $T$; Video file length $D$; History window length $N$; 
  \FOR {k = 1 to D}
  \STATE $\hat{C}_{[t_k,...,t_k+M*T]}$ = ThroughputPred($C_{[t_k-N*T,...,t_k]}$)
  \STATE $\text{QoE}_{max}=0$ 
  \FOR {each possible M-length bitrate trace $a_{[k,...,k+M-1]}$}
  \STATE $\text{QoE}_{[k,...,k+M-1]}=0$
  \FOR{i = 1 to M}
  \STATE $d_i = 0$; $\text{acc\_size} = 0$
  \STATE \textbf{Calculate the transmission delay $d_i$:}
  \FOR{j = 0}
   \IF{$\text{acc\_size}+\hat{C}_{t_k+j*T}*T \geq \rho(a_i)$}
   \STATE $d_i+=\frac{\rho(a_i)-\text{acc\_size}}{\hat{C}_{t_k+j*T}}$
   \STATE \textbf{Break}
   \ENDIF
   \STATE $d_i+= T$; $\text{acc\_size}+=\hat{C}_{t_k+j*T}*T$; j+=1
  \ENDFOR
  \STATE \textbf{Calculate rebuffering:} $\phi_i = (d_i - b_i)_{+}$
  \STATE \textbf{Calculate QoE:} 
  $\text{QoE}_{i} = a_{i,l} - \alpha |a_{i,l} - a_{i-1,l}| -  \beta \phi_{i}$ 
  \STATE \textbf{Update timeline:} $t_{i+1} = t_i + d_i$
  \STATE \textbf{Update buffer occupancy:} $b_{i+1} = (b_i - d_i) + T_{\rm chunk}$
  \ENDFOR
  \STATE $\text{QoE}_{[k,...,k+M-1]}+=\text{QoE}_i$
  \IF{$\text{QoE}_{[k,...,k+M-1]} \geq \text{QoE}_{max}$}
  \STATE $\text{QoE}_{max} = \text{QoE}_{[k,...,k+M-1]}$
  \ENDIF
  \ENDFOR
  \STATE Output the first element $a_{k+1}$ from the bitrate trace that achieves the maximum $\text{QoE}_{[k,...,k+M-1]}$
  \ENDFOR
  \end{algorithmic}
\end{algorithm}

\section{Distortions captured on StreamOptix}

\begin{figure}[!h]
  \centering
    \includegraphics[width=\linewidth]{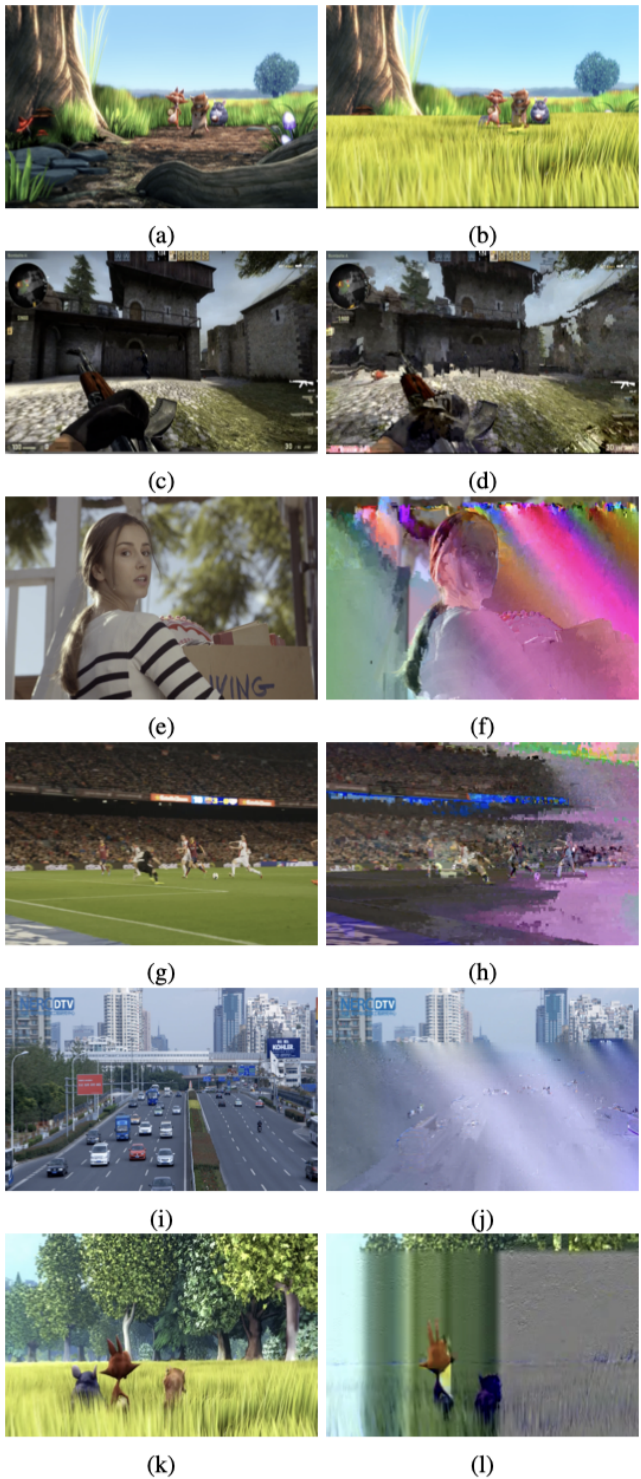}
    % \caption{Image 1}
     \caption{Visualization of the distortion patterns captured in
StreamOptix.}
    \label{fig:1a}
\end{figure}
\subsection{Visualization of video distortions}
Fig. 14 shows different distortions captured when transmitting videos over StreamOptix. In particular, once the video stream is encoded into a bitstream at the application layer, it is then encapsulated into transport blocks (TBs) when it reaches the wireless link. These TBs are subject to random multipath fading and doppler spread within the wireless link, which can hinder their correct decoding at the receiver's end. As a result, this can impact the bitstream decoding process, causing various types of video distortion. Compared with existing trace-driven based ABR emulation platforms (Pensieve), StreamOptix is able to simulate various realistic corruption patterns including reference error (a, b), misalignement (c, d), color artifacts (e, f), block artifacts (g, h), texture loss (i, j), duplication artifacts (k, l). These distortion patterns are influenced not only by the bit error rate of the transmitted video but also by the error location, which tends to be random and unpredictable. Therefore, the most effective way to mitigate these issues is to enhance PHY link adaptation to lower the likelihood of errors. In the future, we could integrate error-specific video recovery methods into our system for further enhancements. As shown in Fig. 1, videos transmitted via StreamOptix more closely resemble corrupted videos in real-world scenarios.

% % use section* for acknowledgment
% \ifCLASSOPTIONcompsoc
%   % The Computer Society usually uses the plural form
%   
% \section*{Acknowledgments}
% The authors would like to thank the referees for their valuable comments.

\ifCLASSOPTIONcaptionsoff
  \newpage
\fi



%\bibliographystyle{IEEEtran}

\bibliographystyle{IEEEtran}
\bibliography{IEEEabrv,Bibliography}


% \end{IEEEbiography}


% % insert where needed to balance the two columns on the last page with
% %\newpage



%\vfill


\end{document}